\documentclass[11pt]{article}

\usepackage[pagebackref,colorlinks]{hyperref}
\usepackage{amsmath} 
\usepackage{amsthm} 
\usepackage{amssymb}	
\usepackage{graphicx} 
\usepackage{multicol} 
\usepackage{multirow}
\usepackage{color}
\usepackage[dvips,letterpaper,margin=1in,bottom=1in]{geometry}
\usepackage[capitalize,noabbrev]{cleveref}

\usepackage[utf8]{inputenc}
\usepackage[english]{babel}

\usepackage{mathtools}

\newtheorem{theorem}{Theorem}[section]

\newtheorem{lemma}[theorem]{Lemma}
\newtheorem{corollary}[theorem]{Corollary}
\newtheorem{proposition}[theorem]{Proposition}

\newcommand{\braket}[2]{\left< #1 \vphantom{#2} \middle| #2 \vphantom{#1} \right>} 

\DeclarePairedDelimiter\rbra{\lparen}{\rparen}
\DeclarePairedDelimiter\sbra{\lbrack}{\rbrack}
\DeclarePairedDelimiter\cbra{\{}{\}}
\DeclarePairedDelimiter\abs{\lvert}{\rvert}
\DeclarePairedDelimiter\Abs{\lVert}{\rVert}

\DeclarePairedDelimiter\ket{\lvert}{\rangle}
\DeclarePairedDelimiter\bra{\langle}{\rvert}

\newcommand{\tr} {\operatorname{tr}}
\newcommand{\poly} {\operatorname{poly}}

\newcommand{\polylog} {\operatorname{polylog}}

\usepackage{algorithm}
\usepackage{algpseudocode}

\usepackage{tabularx}
\usepackage{booktabs}
\usepackage{threeparttable}

\newcommand{\footremember}[2]{%
    \footnote{#2}
    \newcounter{#1}
    \setcounter{#1}{\value{footnote}}%
}

\usepackage{tikz}
\usetikzlibrary{quantikz2}

\title{Optimal Trace Distance and Fidelity Estimations \\ for Pure Quantum States}
\author{
    Qisheng Wang \footremember{1}{Qisheng Wang is with the Graduate School of Mathematics, Nagoya University, Nagoya 464-8602, Japan (e-mail: \href{mailto:QishengWang1994@gmail.com}{\nolinkurl{QishengWang1994@gmail.com}}).}
}
\date{}

\begin{document}

\maketitle

\begin{abstract}
    Measuring the distinguishability between quantum states is a basic problem in quantum information theory.
         In this paper, we develop \textit{optimal} quantum algorithms that estimate both the trace distance and the (square root) fidelity between pure states to within additive error $\varepsilon$ using $\Theta(1/\varepsilon)$ queries to their state-preparation circuits, quadratically improving the long-standing folklore $O(1/\varepsilon^2)$. 
         At the heart of our construction, is an algorithmic tool for quantum square root amplitude estimation, which generalizes the well-known quantum amplitude estimation.
\end{abstract}


\newpage

\tableofcontents
\newpage
\section{Introduction}

    The distinguishability between quantum states is an important topic in quantum information theory. 
    Trace distance and fidelity are the most common measures of the closeness of two quantum states (cf.\ \cite{NC10}). 
    For two mixed quantum states $\rho$ and $\sigma$, the trace distance between them is defined by (cf.\ \cite[Equation (9.11)]{NC10})
    \begin{equation}
        \mathrm{T}\rbra{\rho, \sigma} = \frac 1 2 \tr\rbra{\abs{\rho - \sigma}}.
    \end{equation}
    The fidelity between $\rho$ and $\sigma$ is defined by (cf.\ \cite[Equation (9.53)]{NC10})
    \begin{equation}
        \mathrm{F} \rbra{\rho, \sigma} = \tr\rbra*{\sqrt{\sqrt{\sigma} \rho \sqrt{\sigma}}}.
    \end{equation}
    An alternative definition of fidelity also appears in the literature, defined by (cf.\ \cite{Uhl76,Joz94})
    \begin{equation}
        \mathrm{F}^2 \rbra{\rho, \sigma} = \rbra*{ \tr\rbra*{\sqrt{\sqrt{\sigma} \rho \sqrt{\sigma}}} }^2,
    \end{equation}
    which is equal to the square of $\mathrm{F} \rbra{\rho, \sigma}$. 
    To avoid confusion, throughout this paper, we call $\mathrm{F} \rbra{\rho, \sigma}$ the square root fidelity, and call $\mathrm{F}^2 \rbra{\rho, \sigma}$ the squared fidelity. 

    The trace distance and fidelity indicate how close two quantum states are. 
    When $\rho$ and $\sigma$ are close, the value of $\mathrm{T}\rbra{\rho, \sigma}$ is close to $0$ while the value of $\mathrm{F} \rbra{\rho, \sigma}$ is close to $1$. 
    A quantitative relationship between trace distance and fidelity was given in \cite[Theorem 1]{FvdG99} that 
    \begin{equation}
        1 - \mathrm{F} \rbra{\rho, \sigma} \leq \mathrm{T} \rbra{\rho, \sigma} \leq \sqrt{1 - \mathrm{F}^2 \rbra{\rho, \sigma}}.
    \end{equation}
    When $\rho = \ket{\varphi} \bra{\varphi}$ and $\sigma = \ket{\psi} \bra{\psi}$ are pure quantum states, the relationship between the two measures turns out to be simpler (cf.\ \cite[Equation (9.173)]{Wil13}):
    \begin{equation} \label{eq:relation-td-fi}
        \mathrm{T} \rbra{\ket{\varphi}, \ket{\psi}} = \sqrt{1 - \mathrm{F}^2 \rbra{\ket{\varphi}, \ket{\psi}}},
    \end{equation}
    where
    \begin{equation}
        \mathrm{F} \rbra{\ket{\varphi}, \ket{\psi}} = \abs*{\braket{\varphi}{\psi}}.
    \end{equation}

    An important question is how to estimate the value of the trace distance and fidelity between two quantum states so that we can analyze their closeness quantitatively.
    This is not solely a theoretical question, but also has potential in testing the equivalence of quantum circuits, the validity of quantum devices, the correctness of quantum programs, etc. 

    \subsection{Related Work} \label{sec:related-work}

    The investigation in estimating and testing the closeness of quantum states reveals better understanding of the power and limitations of quantum computing, and they have been shown useful both in theory and in practice. 
    There are two types of quantum input models that are commonly employed:
    \begin{itemize}
        \item Quantum query access model. In this model, quantum query access to quantum unitary oracles that prepare the quantum states to be tested is given.\footnote{When given quantum query access to a quantum unitary oracle $U$, we also assume (by default) that quantum query access to $U^\dag$, controlled-$U$, and controlled-$U^\dag$ is available. For the formal definition, please refer to \cref{sec:def-model}. 
        In addition to the access to $U$, it is necessary to also assume access to the controlled version and the inverse of $U$.
        This is because, when only provided with access to an unknown unitary oracle $U$, implementing controlled-$U$ is impossible \cite{AFCB14}, and implementing $U^\dag$ is challenging \cite{QDS+19}.
        Algorithms in the quantum query access model are particularly useful in the case when the circuit implementation of $U$ is known. 
        In such cases, the circuit implementations of $U^\dag$, controlled-$U$, and controlled-$U^\dag$ can be straightforwardly derived from that of $U$. This scenario is common, as most quantum algorithms are described by their circuit implementations.} 
        Such quantum unitary oracles are also known as the state-preparation circuits.
        The complexity in this model is measured by the number of queries to the oracles (including their inverses and controlled versions), called quantum query complexity.
        Research in this model includes quantum algorithms (e.g., \cite{GL20,vACGN23}) and quantum computational complexity (e.g., \cite{Wat02}). 
        \item Quantum sample access model. In this model, quantum sample access to independent and identical copies of the quantum states to be tested is given. 
        The complexity in this model is measured by the number of samples of the quantum states, called quantum sample complexity. 
        Quantum algorithms in this model include, e.g., \cite{HHJ+17,OW16,BOW19}.
    \end{itemize}

    There are a few approaches for pure quantum states in the literature. 
    One of the earliest approaches is the SWAP test \cite{BCWdW01}, which can be used to estimate the squared fidelity between pure quantum states to within additive error $\varepsilon$ with sample complexity $O\rbra{1/\varepsilon^2}$ or with query complexity $O\rbra{1/\varepsilon}$.
    The SWAP test also implies folklore methods for estimating the trace distance and the square root fidelity between pure quantum states (see \cref{sec:folklore} for further explanations). 
    In \cite{FL11}, the direct squared fidelity estimation for pure states was proposed under the restriction that only Pauli measurements are allowed.
    Another method called entanglement witness (cf.\ \cite[Section 6]{GT09}) can be used for squared fidelity estimation for some specific pure states with few measurements.
    Recently, a distributed quantum algorithm for pure-state squared fidelity estimation with independent and identical samples of quantum states as input was proposed in \cite{ALL22}.

    There are a number of approaches for mixed quantum states. 
    In \cite{BOW19}, the sample complexity of the closeness testing between mixed quantum states of rank $r$ was shown to be $\Theta\rbra{r/\varepsilon^2}$ with respect to trace distance and $\Theta\rbra{r/\varepsilon}$ with respect to square root fidelity. 
    When an $n$-dimensional mixed quantum state is not low-rank, we can just replace $r$ with $n$ in the complexity. 
    In \cite{GL20}, they showed that given the quantum query oracles that prepare the purifications of $n$-dimensional mixed quantum states, the query complexity for the closeness testing between them with respect to trace distance is $O\rbra{n/\varepsilon}$. 
    For square root fidelity estimation, the query complexity was shown to be $\widetilde O\rbra{r^{12.5}/\varepsilon^{13.5}}$ in \cite{WZC+23},\footnote{$\widetilde O\rbra{\cdot}$ suppresses polylogarithmic factors.} and was later improved to $\widetilde O\rbra{r^{6.5}/\varepsilon^{7.5}}$ in \cite{WGL+22} and $\widetilde O\rbra{r^{2.5}/\varepsilon^5}$ in \cite{GP22}; moreover, in \cite{GP22}, they showed that the sample complexity is $\widetilde O\rbra{r^{5.5}/\varepsilon^{12}}$.
    In addition, if two quantum states $\rho$ and $\sigma$ are well-conditioned, i.e., $\rho, \sigma \geq I/\kappa$ for some known $\kappa > 1$, then the query complexity for estimating their fidelity was shown to be $\widetilde O\rbra{\kappa^4/\varepsilon}$ in \cite{LWWZ24}.
    For trace distance estimation, the query complexity was shown to be $\widetilde O\rbra{r^5/\varepsilon^6}$ in \cite{WGL+22}, and was later improved to $\widetilde O\rbra{r/\varepsilon^2}$ \cite{WZ23}; moreover, in \cite{WZ23}, they showed that the sample complexity is $\widetilde O\rbra{r^2/\varepsilon^5}$. 
    In addition to the general approaches, there are also approaches for trace distance estimation in some practical scenarios proposed in \cite{ZRC19,ZR23}. 
    Several variational quantum algorithms for trace distance and fidelity estimations were proposed in \cite{CPCC20,CSZW22,TV21}.

    In computational complexity theory, quantum state discrimination with respect to trace distance (resp.\ fidelity), namely, the decision version of its estimation for mixed quantum states, is known to be $\mathsf{QSZK}$-complete in certain parameter regime \cite{Wat02,Wat09}.
    If the quantum states are guaranteed to be pure, then the quantum state discrimination is $\mathsf{BQP}$-complete \cite{RASW23,WZ23}. 
    Moreover, if the quantum states are prepared by polynomial-size quantum circuits acting on logarithmically many qubits, then the quantum state discrimination is $\mathsf{BQL}$-complete \cite{LGLW23}. 

    \subsection{Main Results}

    As quantum states are usually considered as the quantum analog of classical probability distributions, 
    the estimations of trace distance and (square root) fidelity of quantum states can be seen as quantum analogs of the closeness estimation of classical probability distributions.
    Optimal estimators have been known for closeness measures of probability distributions such as total variation distance and entropy \cite{VV17}.
    However, to the best of our knowledge, we are only aware of the folklore approaches for estimating the trace distance and square root fidelity between pure states which are based on the SWAP test \cite{BCWdW01} and have query complexity $O\rbra{1/\varepsilon^2}$ for additive error $\varepsilon$ (see \cref{lemma:folklore-td} and \cref{lemma:folklore-sqrt-fi}).\footnote{In the special case where the amplitudes of the pure states are guaranteed to be real numbers, a quantum query algorithm for square root fidelity estimation with query complexity $O\rbra{1/\varepsilon}$ was implied in \cite[Theorem 2.1]{KLLP19}. However, their method does not directly apply to the general complex-valued case.} 

    In this paper, we propose quantum query algorithms for estimating the trace distance and square root fidelity between pure quantum states with query complexity $\Theta\rbra{1/\varepsilon}$, assuming that their state-preparation circuits are given as quantum unitary oracles. 
    Moreover, we show that they are \textit{optimal} by providing matching lower bounds on the quantum query complexity of pure-state trace distance and square root fidelity estimations. 
    The input model adopted in our quantum algorithms is commonly employed in the study of quantum query algorithms regarding quantum states; for example, the query complexity of quantum state tomography was studied in \cite{vACGN23}.
    See \cref{sec:input-model} for the formal definition of this input model.
    
    We state our main results in the following theorem. 

    \begin{theorem} [Optimal pure-state trace distance and square root fidelity estimations] \label{thm:td-intro} \
        \begin{itemize}
            \item \textbf{Upper Bounds} (\cref{thm:td} and \cref{thm:fi} combined): There is a quantum query algorithm that estimates the trace distance and the square root fidelity between two pure quantum states to within additive error $\varepsilon$ with probability at least $2/3$ with query complexity $O\rbra{1/\varepsilon}$. 
            \item \textbf{Lower Bounds} (\cref{thm:qlower-td} and \cref{thm:qlb-sqrt-fi} combined): Any quantum query algorithm that estimates the trace distance or the square root fidelity between two pure quantum states to within additive error $\varepsilon$ with probability at least $2/3$ has query complexity $\Omega\rbra{1/\varepsilon}$. 
        \end{itemize}
    \end{theorem}

    In a nutshell, \cref{thm:td-intro} means that the query complexities for both pure-state trace distance estimation and pure-state square root fidelity estimation are $\Theta\rbra{1/\varepsilon}$, which quadratically improve the folklore results of $O\rbra{1/\varepsilon^2}$ that are based on the SWAP test (see \cref{lemma:folklore-td} and \cref{lemma:folklore-sqrt-fi}).
    Moreover, we show that $\Omega \rbra{1/\varepsilon}$ queries are necessary for pure-state trace distance and square root fidelity estimations, meaning that our quantum algorithm in \cref{thm:td-intro} is optimal only up to a constant factor.
    We also note that the quantum query algorithm given in \cref{thm:td-intro} can also be used to reproduce the pure-state squared fidelity estimation that is originally obtained by combining the SWAP test \cite{BCWdW01} and quantum amplitude estimation \cite{BHMT02} (see \cref{thm:sqr-fi}). 
    As the optimal estimation for pure-state squared fidelity based on the SWAP test is already known and has a wide range of applications, the readers may wonder how the pure-state trace distance and square root fidelity estimations given in \cref{thm:td-intro} can be useful.
    In \cref{sec:app}, we present concrete applications in which trace distance and square root fidelity estimations perform significantly better than the squared fidelity estimation based on the SWAP test. 

\begin{table*}[!htp]
\centering
\caption{Quantum query and sample complexities for pure-state closeness estimations.}
\label{tab:cmp}
\begin{tabular}{cccc}
\toprule
Bounds             & Trace Distance                                                                                                      & Square Root Fidelity                                                                                                      & Squared Fidelity                                    \\ \midrule
Query Upper Bound  & \begin{tabular}[c]{@{}c@{}}$O(1/\varepsilon^2)$ folklore\\ $O(1/\varepsilon)$ \cref{thm:td} \end{tabular} & \begin{tabular}[c]{@{}c@{}}$O(1/\varepsilon^2)$ folklore\\ $O(1/\varepsilon)$ \cref{thm:fi} \end{tabular} & $O(1/\varepsilon)$ \cite{BCWdW01,BHMT02}               \\ \midrule
Query Lower Bound  & $\Omega(1/\varepsilon)$ \cref{thm:qlower-td}                                                                   & $\Omega(1/\varepsilon)$ \cref{thm:qlb-sqrt-fi}                                                                    & $\Omega(1/\varepsilon)$ \cite{BBC+01,NW99}      \\ \midrule
Sample Upper Bound & $O(1/\varepsilon^4)$ folklore                                                                     & $O(1/\varepsilon^4)$ folklore                                                                     & $O(1/\varepsilon^2)$ \cite{BCWdW01}        \\ \midrule
Sample Lower Bound & $\Omega(1/\varepsilon^2)$ \cref{thm:slb-td}                                                                  & $\Omega(1/\varepsilon^2)$ \cref{thm:slb-sqrt-fi}                                                                  & $\Omega(1/\varepsilon^2)$ \cite{ALL22} \\ \bottomrule
\end{tabular}
\end{table*}

    In \cref{tab:cmp}, we compare the quantum query complexity and sample complexities for pure-state trace distance, square root fidelity, and squared fidelity estimations.
    For completeness, we provide folklore approaches concerning both query complexity and sample complexity in \cref{sec:folklore} (with the complexity for squared fidelity estimation given in \cref{lemma:folklore-sqr-fi}), and quantum sample lower bounds in \cref{sec:qslb}. 

    To better illustrate our results, we first mention its potential applications in \cref{sec:app}, and then introduce our techniques in \cref{sec:tech}: upper bounds in \cref{sec:tech-ub} and lower bounds in \cref{sec:tech-lb}. 
    The hardness in computational complexity theory will be discussed in \cref{sec:hardness}.
    Finally, we will give a brief discussion in \cref{sec:discussion}.

    \subsection{Applications} \label{sec:app}

    We present three concrete and representative applications of \cref{thm:td-intro} that could be of broad interest. 
    In the following, in addition to trace distance, we will also use the distance-like measure \textit{infidelity} to quantify the closeness between quantum states, defined by
    $\mathrm{I}\rbra{\rho, \sigma} = 1 - \mathrm{F}\rbra{\rho, \sigma}$ (note that the infidelity is defined in terms of square root fidelity).

    \subsubsection{Quantum State Tomography}
    Quantum state tomography \cite{HHJ+17,OW16} is a fundamental problem in quantum computing. 
    The task is to estimate an unknown quantum state to certain precision. 
    Existing approaches mainly measure errors in trace distance and infidelity. 
    In particular, the tomography of a $d$-dimensional pure quantum state can be done optimally by using $\Theta\rbra{d/\varepsilon^2}$ samples to error $\varepsilon$ in trace distance \cite{OW16}, and by using $\widetilde \Theta\rbra{d/\delta}$ samples to error $\delta$ in infidelity \cite{HHJ+17}. 
    When testing a concrete implementation for quantum state tomography, one may wish to see if its output is correct (i.e., with high enough precision) provided that the input state for testing is known in advance. 
    A simple way is to check if the trace distance (resp. infidelity) between the input and output states is small enough. 
    
    Suppose that the input state $\ket{\psi_{\textup{in}}}$ can be prepared by a known quantum circuit of size $T_{\textup{in}}$, and the classical description of the output state $\ket{\psi_{\textup{out}}}$ is given such that $\ket{\psi_{\textup{out}}}$ can be prepared with time complexity $\polylog\rbra{d} = \widetilde O\rbra{1}$.\footnote{It was shown in \cite{ZLY22,STY+23,Ros21,YZ23} that any $d$-dimensional (pure) quantum state can be prepared by a quantum circuit of depth $\Theta\rbra{\log\rbra{d}}$. Moreover, if the classical description of a pure quantum state is stored in QRAM (quantum random access memory) \cite{GLM08} equipped with a specific tree data structure \cite{KP17}, then the pure state can be prepared with time complexity $\polylog\rbra{d}$ on a quantum computer. It is noted that the controlled version of the state-preparation and its inverse can also be implemented with the same time complexity.} 
    We can estimate the trace distance and the infidelity between $\ket{\psi_{\textup{in}}}$ and $\ket{\psi_{\textup{out}}}$ on a quantum computer through the pure-state trace distance and square root fidelity estimations given in \cref{thm:td-intro}.
    To assess with high probability whether the tomography is successful, one can employ the known state $\ket{\psi_{\textup{in}}}$ as the input for the tomography. Subsequently, the tomography produces the resulting state $\ket{\psi_{\textup{out}}}$ represented as classical data.
    Following this approach, $\ket{\psi_{\textup{out}}}$ can be efficiently prepared after certain preprocessing, e.g., storing its amplitudes in a QRAM.
    Finally, the assessment is completed by estimating the closeness between $\ket{\psi_{\textup{in}}}$ and $\ket{\psi_{\textup{out}}}$.
    The time complexity is given as follows. 
    \begin{proposition} [Testing quantum state tomography] \label{prop:test-qst}
        With the assumptions mentioned above, we can estimate the trace distance and the infidelity between $\ket{\psi_{\textup{in}}}$ and $\ket{\psi_{\textup{out}}}$ to within additive error $\varepsilon$ on a quantum computer with time complexity $\widetilde O\rbra{T_{\textup{in}} / \varepsilon}$. 
    \end{proposition}
    It is worth noting that the approach in \cref{prop:test-qst} has a linear dependence on $\varepsilon$. 
    By comparison, any approach based on the squared fidelity estimation will only result in a time complexity of $\widetilde O\rbra{T_{\textup{in}} / \varepsilon^2}$.
    It can be seen that both pure-state trace distance and square root fidelity estimations are useful in testing quantum state tomography, with a quadratic speedup in the parameter $\varepsilon$ over the prior best approach based on squared fidelity estimation (i.e., the SWAP test). 
    
    \subsubsection{Quantum State Discrimination}
    We consider two situations of quantum state discrimination (cf.\ \cite{Che00,BC09,BK15}). 

    \paragraph{Minimum-error quantum hypothesis testing.}
    In quantum hypothesis testing, a pure quantum state $\ket{\phi}$ is given such that the two cases hold with equal probability: (i) $\ket{\phi} = \ket{\varphi}$, and (ii) $\ket{\phi} = \ket{\psi}$, where $\ket{\varphi}$ and $\ket{\psi}$ are two known pure quantum states.
    A basic problem is to find the minimum error probability for distinguishing the two cases.
    By the Helstrom-Holevo bound \cite{Hel67,Hol73}, the minimum error is given in terms of trace distance:
    \begin{equation}
        p_{\textup{err}} = \frac 1 2 - \frac 1 2 \mathrm{T}\rbra{\ket{\varphi}, \ket{\psi}}.
    \end{equation}
    However, it is necessary to know the minimum error for specific circumstances in practice. 
    For the two known quantum states $\ket{\varphi}$ and $\ket{\psi}$, once we have their state-preparation circuits, we can compute the probability $p_{\textup{err}}$ efficiently on a quantum computer through the pure-state trace distance estimation given in \cref{thm:td-intro}.
    \begin{proposition} [Computing the minimum error of quantum hypothesis testing] \label{corollary:min-err-qht}
        Suppose that we are given two quantum circuits of size $T_{\varphi}$ and $T_{\psi}$ that prepare the two known $d$-dimensional pure states $\ket{\varphi}$ and $\ket{\psi}$, respectively. 
        Then, we can compute the minimum error probability of the quantum hypothesis testing for $\ket{\varphi}$ and $\ket{\psi}$ to within additive error $\varepsilon$ on a quantum computer with time complexity $\widetilde O\rbra{\rbra{T_{\varphi}+T_{\psi}}/\varepsilon}$. 
    \end{proposition}
    It is worth noting that the algorithm in \cref{corollary:min-err-qht} achieves a linear dependence on the precision $\varepsilon$, and prior to this, the folklore approach based on squared fidelity estimation (i.e., the SWAP test) will only result in a quadratic dependence on $\varepsilon$. 

    \paragraph{Sample complexity for quantum state discrimination.}
    When considering the sample complexity of quantum state discrimination, one is given with independent and identical samples of an unknown quantum state $\ket{\phi}$. 
    The task is to determine whether $\ket{\phi} = \ket{\varphi}$ or $\ket{\phi} = \ket{\psi}$ for two known states $\ket{\varphi}$ and $\ket{\psi}$, promised that it is in either case. 
    The sample complexity of quantum state discrimination, denoted as $\mathsf{S}\rbra{\ket{\varphi}, \ket{\psi}}$, is the minimum positive integer $S$ such that with high probability (say, greater than $2/3$) one can distinguish the two cases from the state $\ket{\phi}^{\otimes S}$. 
    It is known (cf.\ \cite{Wot81,HHJ+17}) that the growth of the sample complexity can be fully characterized by (the inverse of) infidelity:
    \begin{equation} \label{eq:bound-s}
        C_1/\mathrm{I}\rbra{\ket{\varphi}, \ket{\psi}} \leq \mathsf{S}\rbra{\ket{\varphi}, \ket{\psi}} \leq C_2/\mathrm{I}\rbra{\ket{\varphi}, \ket{\psi}}
    \end{equation}
    for some constant $C_1, C_2 > 0$. 
    For specific circumstances, one may wish to find reasonable bounds on the sample complexity. 
    Suppose that we know the quantum circuits that prepare the two known pure states $\ket{\varphi}$ and $\ket{\psi}$.
    By the pure-state square root fidelity estimation in \cref{thm:td-intro}, we can estimate $1/\mathrm{I}\rbra{\ket{\varphi}, \ket{\psi}}$ with relative error $\varepsilon$ on a quantum computer, and thus obtain upper and lower bounds on $\mathsf{S}\rbra{\ket{\varphi}, \ket{\psi}}$ with relative error $\varepsilon$ according to \cref{eq:bound-s}.

    \begin{proposition} [Bounding the sample complexity for quantum state discrimination] \label{corollary:sc-qsd}
        Suppose that we are given two quantum circuits of size $T_{\varphi}$ and $T_{\psi}$ that prepare the two known $d$-dimensional pure states $\ket{\varphi}$ and $\ket{\psi}$, respectively, with $B \leq \mathrm{I}\rbra{\ket{\varphi}, \ket{\psi}}$. 
        Then, we can estimate $1/\mathrm{I}\rbra{\ket{\varphi}, \ket{\psi}}$ with relative error $\varepsilon$ on a quantum computer with time complexity $\widetilde O\rbra{\rbra{T_{\varphi}+T_{\psi}}/\rbra{B\varepsilon}}$. 
    \end{proposition}

    Similar to \cref{corollary:min-err-qht}, the algorithm in \cref{corollary:sc-qsd} achieves a linear dependence on the precision $\varepsilon$.
    By comparison, the prior best approach is based on squared fidelity estimation (i.e., the SWAP test), which will only result in a quadratic dependence on $\varepsilon$. 

    \subsection{Techniques} \label{sec:tech}

    \subsubsection{Upper Bounds} \label{sec:tech-ub}

    A straightforward approach for pure-state trace distance estimation is to first estimate the fidelity between pure quantum states and then compute the trace distance by \cref{eq:relation-td-fi} using arithmetic operations, which will result in a query complexity of $O\rbra{1/\varepsilon^2}$ (see \cref{sec:folklore} for detailed explanations).
    Here, the difficulty comes from the numerical stability when taking the square root (see \cref{prop:sqrt-stable}). 
    A possible way to improve this approach is to skip these square root operations. 
    In fact, this can be accomplished through the very nature of quantum computing.

    At a high level, square roots are essentially ubiquitous in the basis of quantum computing.
    As an illustration, when we measure a pure quantum state $\ket{\psi} = \alpha \ket{0} + \beta \ket{1}$ in the computational basis, the outcome will be $0$ with probability $p_0 = \abs{\alpha}^2$ and $1$ with probability $p_1 = \abs{\beta}^2$. 
    As shown in this example, \textit{the absolute values of the amplitudes actually store the square root of the probabilities}, as $\abs{\alpha} = \sqrt{p_0}$ and $\abs{\beta} = \sqrt{p_1}$. 
    The quantum amplitude estimation \cite{BHMT02} allows us to estimate the probability $p_0$ to within additive error $\varepsilon$ with query complexity $O\rbra{1/\varepsilon}$, which means that we can estimate $\sqrt{p_0}$ to within additive error $\sqrt{\varepsilon}$ if we directly take the square root of the estimated value of $p_0$ (the numerical error is guaranteed by \cref{prop:sqrt-stable}). 
    Regarding this, a natural question is:
    \begin{equation} \tag{$*$} \label{eq:question}
    \begin{gathered}
    \textit{Can we directly learn the square root of probabilities} \\
    \textit{rather than just the probabilities themselves} \\
    \textit{on a quantum computer?}
    \end{gathered}
    \end{equation}
    
    If this could be done efficiently, the task of pure-state trace distance estimation (as well as pure-state fidelity estimation) would be a candidate application that is noteworthy. 
    To see this, suppose that $U_{\varphi}$ and $U_{\psi}$ are two quantum unitary operators that prepare two $k$-qubit pure quantum states $\ket{\varphi} = U_{\varphi} \ket{0}$ and $\ket{\psi} = U_{\psi} \ket{0}$, respectively. 
    Then, we can construct a unitary operator 
    \begin{equation} \label{eq:def-W-intro}
    \begin{aligned}
        W = & \rbra*{X_{\mathsf{A}} \otimes \ket{0}_{\mathsf{B}} \bra{0} + \sum_{j = 1}^{2^k - 1} I_{\mathsf{A}} \otimes \ket{j}_{\mathsf{B}} \bra{j}} \cdot \rbra*{I_{\mathsf{A}} \otimes \rbra*{U_{\varphi}^\dag U_{\psi}}_{\mathsf{B}}},
    \end{aligned}
    \end{equation}
    which prepares a pure quantum state with the trace distance $\mathrm{T} \rbra{\ket{\varphi}, \ket{\psi}}$ encoded in one of its amplitudes.
    More precisely, 
    \begin{equation}
        W\ket{0}_{\mathsf{A}}\ket{0}_{\mathsf{B}} = \sqrt{p} \ket{0}_{\mathsf{A}} \ket{\phi_0}_{\mathsf{B}} + \sqrt{1 - p} \ket{1}_{\mathsf{A}} \ket{\phi_1}_{\mathsf{B}}
    \end{equation}
    for some normalized pure quantum states $\ket{\phi_0}$ and $\ket{\phi_1}$, where $\sqrt{p} = \mathrm{T} \rbra{\ket{\varphi}, \ket{\psi}}$.
    See the proof of \cref{thm:td} for more details about pure-state trace distance estimation. 
    The aforementioned approach for pure-state trace distance estimation can also be adjusted to pure-state square root fidelity estimation effortlessly by noting that $\sqrt{1 - p} = \mathrm{F} \rbra{\ket{\varphi}, \ket{\psi}}$. 

    \paragraph{Square Root Amplitude Estimation.}
    
    Lastly, we give a positive answer to the question (\ref{eq:question}) by providing a quantum query algorithm for square root amplitude estimation, which generalizes the well-known quantum amplitude estimation proposed in \cite{BHMT02}. 
    As the name implies, the square root amplitude estimation allows us to directly estimate the square root of the probability $p_0$ of the measurement outcome $0$, whereas the original amplitude estimation \cite{BHMT02} only allows us to estimate the value of $p_0$ itself. 
    This algorithmic tool is formally stated as follows, which, we believe, could be used as a basic subroutine in numerous applications in future research. 

    \begin{theorem} [Square root amplitude estimation, \cref{thm:sqrt-ampl-est} restated] \label{thm:sqrt-amp-est}
        Suppose that $U$ is a quantum unitary oracle such that
        \begin{equation}
        U\ket{0}_{\mathsf{A}}\ket{0}_{\mathsf{B}} = \sqrt{p} \ket{0}_{\mathsf{A}}\ket{\phi_0}_{\mathsf{B}} + \sqrt{1-p} \ket{1}_{\mathsf{A}}\ket{\phi_1}_{\mathsf{B}},
        \end{equation}
        where $\ket{\phi_0}$ and $\ket{\phi_1}$ are normalized pure states. 
        Then, for every $\varepsilon \in \rbra{0, 1}$, we can estimate the value of $\sqrt{p}$ to within additive error $\varepsilon$, using $O\rbra{1/\varepsilon}$ queries to $U$.
    \end{theorem}

    The square root amplitude estimation given in \cref{thm:sqrt-amp-est} can be quadratically faster than the original amplitude estimation in \cite{BHMT02} (which will result in query complexity $O\rbra{1/\varepsilon^2}$) when the goal is to estimate the square root probability $\sqrt{p}$ (see explanations in \cref{sec:ampl-est}). 
    In spite of the quadratic speedup, our algorithm for square root amplitude estimation essentially builds on the framework of the original amplitude estimation \cite{BHMT02}, with a simple observation.\footnote{In \cite[Appendix C]{dW19}, a similar observation was noted: 
    the differences between $\sqrt{p}$, $\sin\rbra{\sqrt{p}}$, and $\arcsin\rbra{\sqrt{p}}$ are negligible when $p \ll 1$.
    This serves as a hint for Exercise 8 in Chapter 7 of \cite{dW19}.}
    To introduce the idea, let us review the main process of amplitude estimation \cite{BHMT02}. 
    Let
    \begin{equation}
    \begin{aligned}
        Q = 
        & -U \cdot \rbra*{I_{\mathsf{A}} \otimes I_{\mathsf{B}} - 2 \ket{0}_{\mathsf{A}}\bra{0} \otimes \ket{0}_{\mathsf{B}}\bra{0}} \cdot U^\dag \cdot \rbra*{I_{\mathsf{A}} \otimes I_{\mathsf{B}} - 2\ket{0}_{\mathsf{A}}\bra{0} \otimes I_{\mathsf{B}}},
    \end{aligned}
    \end{equation}
    which uses $2$ queries to $U$ (and $U^\dag$). 
    Then, we can find a pair of eigenvectors $\ket{\psi_{\pm}}$ of $Q$ with eigenvalues $e^{\pm i 2 \theta_p}$, where $\theta_p = \arcsin\rbra{\sqrt{p}}$.
    After we obtain an estimate $\tilde{\theta}_p$ of $\theta_p$ to within additive error $\varepsilon$ by quantum phase estimation \cite{Kit95}, the value of $\sin^2\rbra{\tilde \theta_p}$ is then an estimate of $p$ to within additive error $\Theta\rbra{\varepsilon}$. 
    Our observation in addition to this is that $\abs{\sin\rbra{\tilde \theta_p}}$ is an estimate of $\sqrt{p}$ to within additive error $\Theta\rbra{\varepsilon}$ by noting that the function $\abs{\sin\rbra{\cdot}}$ is $O\rbra{1}$-Lipschitz. 
    With this observation, we can estimate $\sqrt{p}$ to within additive error $\varepsilon$ with query complexity $O\rbra{1/\varepsilon}$, which is the same (up to a constant factor) query complexity for estimating $p$ to the same precision. 
    
    Moreover, the square root amplitude estimation can reproduce the query complexity for amplitude estimation for unknown $p$; and this is why we previously said that the former generalizes the latter. 
    The reduction is simple: first obtain an estimate $\tilde x$ of $\sqrt{p}$ to within additive error $\varepsilon/2$ by the square root amplitude estimation with query complexity $O\rbra{1/\varepsilon}$, and then return $\tilde x^2$ as the estimate of $p$. See \cref{sec:sqrt-ampl-est-algo} for more details. 

    \subsubsection{Lower Bounds} \label{sec:tech-lb}

    Our quantum query lower bound for pure-state trace distance estimation is obtained by a reduction from distinguishing probability distributions, with the quantum query lower bound $\Omega\rbra{1/d_{\textup{H}}\rbra{p, q}}$ for distinguishing two probability distributions $p$ and $q$ given in \cite{Bel19}, where $d_{\textup{H}}\rbra{p, q}$ is the Hellinger distance. 
    In our reduction, we employ a pair of $n$-dimensional probability distributions $p^\pm$ defined by
    \begin{equation} \label{eq:def-prob-distri}
        p^\pm \rbra{j} = \frac{1 \pm \rbra{-1}^j 2 \varepsilon}{n},
    \end{equation}
    with their Hellinger distance upper bounded by $d_{\textup{H}}\rbra{p^+, p^-} \leq 2\varepsilon$. 
    Then, we consider the problem of distinguishing the two probability distributions $p^+$ and $p^-$ that are encoded in two quantum unitary oracles $U_{p^{\pm}}$ such that
    \begin{equation}
        U_{p^{\pm}} \ket{0} = \ket{\psi^\pm} = \sum_{j \in \sbra{n}} \sqrt{p^\pm\rbra{j}} \ket{j},
    \end{equation}
    where $\sbra{n} = \cbra{0, 1, 2, \dots, n - 1}$.
    This is essentially a problem of distinguishing pure quantum states: given an unknown state $\ket{\psi}$, determine whether it is $\ket{\psi^+}$ or $\ket{\psi^-}$, promised that it is in either case. 
    Actually, this can be done by conditioning on whether the trace distance $\mathrm{T}\rbra{\ket{\psi}, \ket{\psi^+}}$ between the tested state $\ket{\psi}$ and the known state $\ket{\psi^+}$ is $0$ (in which case $\ket{\psi} = \ket{\psi^+}$) or $\mathrm{T}\rbra{\ket{\psi^+}, \ket{\psi^-}} = 2\varepsilon$ (in which case $\ket{\psi} = \ket{\psi^-}$). 
    Since any estimate of $\mathrm{T}\rbra{\ket{\psi}, \ket{\psi^+}}$ to within additive error $\varepsilon$ works for the distinguishing problem, any quantum query algorithm for pure-state trace distance estimation to within additive error $\varepsilon$ requires query complexity $\Omega\rbra{1/d_{\textup{H}}\rbra{p^+, p^-}} = \Omega\rbra{1/\varepsilon}$.
    See \cref{thm:qlower-td} for more details. 
    
    A matching quantum query lower bound $\Omega\rbra{1/\varepsilon}$ for pure-state squared fidelity estimation can be obtained similarly, though it was previously known from the quantum query lower bound for quantum counting \cite{BBC+01,NW99}; moreover, it also implies a matching quantum query lower bound for pure-state square root fidelity estimation. 
    See \cref{thm:qlower-sqr-fi} and \cref{thm:qlb-sqrt-fi} for more details. 
    
    We note that the probability distributions defined in \cref{eq:def-prob-distri} were ever employed in proving lower bounds for testing the uniformity of probability distributions in \cite{Pan08} and \cite{LWL23} on classical sample complexity and quantum query complexity, respectively.
    Moreover, such probability distributions were also adapted to proving quantum sample lower bounds for testing the uniformity of mixed quantum states \cite{OW21}. 
    Here, we note the difference between the proof of \cite{OW21} and ours:
    the probability distributions are related to the amplitudes of the pure quantum states in our proof, while they are related to the eigenvalues of the mixed quantum states in the proof of \cite{OW21}.

    \paragraph{Extensions.}
    
    We also extend this method to proving quantum sample lower bounds for these tasks. As they are not the main subject of this paper, we put them in \cref{sec:qslb}. 
    Nevertheless, it is worth noting that the quantum sample lower bound $\Omega\rbra{1/\varepsilon^2}$ for pure-state trace distance estimation (see \cref{thm:slb-td}) is new, in spite of the quantum sample lower bound $\Omega\rbra{r/\varepsilon^2}$ for the closeness testing of mixed quantum states of rank $r$ with respect to trace distance given in \cite{BOW19,OW21}. 
    This is because their proofs do not cover the case of $r = 1$ (i.e., pure quantum states) due to the use of the probability distributions in \cref{eq:def-prob-distri} (which requires $r \geq 2$). 
    
    \subsection{Hardness} \label{sec:hardness}

    In \cref{tab:cmp}, we present lower bounds for various scenarios, illustrating the fine-grained amount of resources that are required to solve the pure-state estimations of trace distance, square root fidelity, and squared fidelity.
    From the perspective of computational complexity theory, these estimation tasks are known to be $\mathsf{BQP}$-hard. 
    For completeness, we include the hardness results presented in \cite{RASW23} and \cite{WZ23} as follows. 

    \begin{theorem} [$\mathsf{BQP}$-completeness, adapted from {\cite[Theorem 12]{RASW23}} and {\cite[Theorem IV.1]{WZ23}}] \label{thm:hardness}
        Given the classical description of two quantum circuits of size $\poly\rbra{n}$ that prepare two $n$-qubit pure quantum states respectively, it is $\mathsf{BQP}$-complete to determine whether the trace distance, the square root fidelity, or the squared fidelity between the two pure quantum states is less than $\alpha$ or greater than $\beta$ for any $2^{-\poly\rbra{n}} \leq \alpha < \beta \leq 1 - 2^{-\poly\rbra{n}}$ with $\beta - \alpha \geq 1/\poly\rbra{n}$.
    \end{theorem}

    If we remove the restriction that the quantum states are pure, the estimation of trace distance, square root fidelity, and squared fidelity (between mixed quantum states) is known to be $\mathsf{QSZK}$-complete for certain regime of $\alpha$ and $\beta$ \cite{Wat02,Wat09}.
    In addition, if the mixed quantum states are low-rank, i.e., they are of rank $r$ where $r = \polylog \rbra{n}$, these estimation tasks fall directly into $\mathsf{BQP}$. 
    In other words, (the decision versions of) the estimations of trace distance, square root fidelity, and squared fidelity between low-rank quantum states are $\mathsf{BQP}$-complete, where the $\mathsf{BQP}$-hardness is due to \cref{thm:hardness} while the $\mathsf{BQP}$-containment is due to the polynomial-time quantum algorithms proposed in \cite{WZC+23,WGL+22,GP22,WZ23}. 

    \subsection{Discussion} \label{sec:discussion}

    In this paper, we provide \textit{optimal} quantum query algorithms for pure-state trace distance and square root fidelity estimations, which are obtained via the quantum algorithmic tool --- square root amplitude estimation given in \cref{thm:sqrt-amp-est}.
    Our results, together with prior results concerning squared fidelity \cite{BCWdW01,BBC+01,NW99}, reveal that the quantum query complexity for pure-state ``closeness'' estimation is $\Theta\rbra{1/\varepsilon}$. 
    However, there are still gaps between the upper and lower bounds on the quantum sample complexity for these estimation tasks, except for the $\Theta\rbra{1/\varepsilon^2}$ sample complexity for pure-state squared fidelity estimation due to \cite{BCWdW01} and \cite{ALL22}.
    An important problem is to close the gap between the upper and lower bounds on the sample complexity.

    As mentioned in \cref{sec:related-work}, there are a series of works towards testing and learning the closeness of mixed quantum states.
    Except for the sample complexity of quantum state certification with respect to both trace distance and square root fidelity \cite{BOW19}, many of them (either in query complexity or sample complexity) are far from being optimal.\footnote{We are only aware of a quantum query algorithm in {\cite{LWWZ24}} that estimates the fidelity between well-conditioned mixed quantum states, achieving an almost optimal dependence on the additive error $\varepsilon$.}
    It would be interesting to develop new techniques for improving these upper and lower bounds. 

    Not only limited to the tasks considered in this paper, we hope that the quantum square root amplitude estimation given in \cref{thm:sqrt-amp-est} brings new ideas to quantum computing and that it could be used as a subroutine in the design of quantum algorithms. 

    \subsection{Organization of This Paper}

    In \cref{sec:quantum-query-model}, we introduce the quantum query model, especially the input model for pure quantum states. 
    In \cref{sec:sqrt-ampl-est}, we provide an efficient quantum query algorithm for square root amplitude estimation. 
    Then, in \cref{sec:td-est}, using the square root amplitude estimation, we present an optimal quantum query algorithm for pure-state trace distance and square root fidelity estimations, and reproduce the pure-state squared fidelity estimation. 
    Finally, in \cref{sec:quantum-lower-bounds}, lower bounds on the quantum query complexity for pure-state trace distance and square root fidelity estimations are proved.

    \section{Quantum Query Model} \label{sec:quantum-query-model}

    In this section, we briefly introduce the quantum query model, especially the input query model for pure quantum states that is used in this paper.

    \subsection{The General Model} \label{sec:def-model}

    In quantum computing, a quantum unitary oracle $U$ is a unitary operator on a finite-dimensional Hilbert space such that $U^\dag U = U U^\dag = I$, where $I$ is the identity operator and $U^\dag$ is the Hermitian conjugate of $U$. 
    A quantum query algorithm $\mathcal{A}$ using $Q$ queries to quantum oracle $U$ is described by a quantum circuit composed of quantum gates and the quantum oracle $U$ (and its inverse $U^\dag$), namely, 
    \begin{equation}
    \mathcal{A} = G_Q \cdot U_Q \cdot \cdots \cdot G_2 \cdot U_2 \cdot G_1 \cdot U_1 \cdot G_0,
    \end{equation}
    where $G_j$ is a quantum gate composed of elementary one- and two-qubit quantum gates that does not depend on $U$, and $U_j$ is either (controlled-)$U$ or (controlled-)$U^\dag$.
    Here, $Q$ is called the (quantum) query complexity of $\mathcal{A}$. 
    To run the algorithm $\mathcal{A}$, we first prepare a (multi-qubit) pure quantum state $\ket{0}$, and then apply the unitary operators $G_0, U_1, G_1, \dots, U_Q, G_Q$ one by one in this order.
    After that, the pure quantum state becomes 
    \begin{equation}
        \mathcal{A} \ket{0} = G_Q  U_Q  \cdots  G_2  U_2  G_1  U_1  G_0 \ket{0}.
    \end{equation}
    To fetch (classical) information from the execution of quantum query algorithm $\mathcal{A}$, we usually make a quantum projective measurement $M = \cbra{P_m}$ in the computational basis of (a subset of) the qubits that $\mathcal{A}$ acts on. 
    Here, we note that $P_m$ are projectors, i.e., $P_m^\dag = P_m$ and $P_m^2 = P_m$, and they satisfy the completeness equation $\sum_m P_m = I$. 
    Then, $\mathcal{A}$ outputs $m$ with probability $\Abs{P_m \mathcal{A}\ket{0} }^2$.

    \subsection{The Input Model for Pure Quantum States} \label{sec:input-model}

    For quantum query algorithms regarding pure quantum states, we assume that their state-preparation circuits are given as quantum unitary oracles. 
    Strictly speaking, quantum query access to a pure quantum state $\ket{\psi}$ means a quantum unitary oracle $U_\psi$ such that 
    \begin{equation}
        \ket{\psi} = U_\psi \ket{0}. 
    \end{equation}
    This input model was employed in quantum state tomography in \cite{vACGN23}, where they show that the quantum query complexity of obtaining an $\varepsilon$-$\ell_2$-approximation of an $n$-dimensional pure quantum state is $\widetilde \Theta\rbra{n/\varepsilon}$. 
    This input model was also employed as an adapter for problems that are not directly related to pure quantum states, e.g., encoding probability distributions (cf.\ \cite{BHH11}) and preparing purifications of mixed quantum states (cf.\ \cite{GL20}). 

    \section{Square Root Amplitude Estimation} \label{sec:sqrt-ampl-est}

    In this section, we provide a quantum query algorithm for square root amplitude estimation (see \cref{thm:sqrt-ampl-est}), which generalizes the quantum amplitude estimation proposed in \cite{BHMT02}. 
    For a unitary operator $U\ket{0}_{\mathsf{A}}\ket{0}_{\mathsf{B}} = \sqrt{p} \ket{0}_{\mathsf{A}}\ket{\phi_0}_{\mathsf{B}} + \sqrt{1-p} \ket{1}_{\mathsf{A}}\ket{\phi_1}_{\mathsf{B}}$ with $p \in \sbra{0, 1}$, the quantum amplitude estimation in \cite{BHMT02} allows us to estimate the value of $p$ (to within an additive error). 
    In comparison, the square root amplitude estimation provided in \cref{thm:sqrt-ampl-est} allows us to estimate the square root of $p$, which is more efficient than just estimating the value of $p$ simply by amplitude estimation and then computing the square root of the estimated value.

    This section is organized as follows. 
    In \cref{sec:ampl-est}, we recall the quantum amplitude estimation proposed in \cite{BHMT02} and analyze the complexity if we simply use it to estimate the value of $\sqrt{p}$.
    In \cref{sec:sqrt-ampl-est-algo}, we present the quantum query algorithm for square root amplitude estimation, and we additionally show how the square root amplitude estimation can reproduce the quantum amplitude estimation in \cite{BHMT02} when no prior knowledge of $p$ is known.

    \subsection{Amplitude Estimation} \label{sec:ampl-est}

    The well-known quantum query algorithm for amplitude estimation was proposed in \cite{BHMT02}.
    We recall a simple version of it when no prior knowledge of $p$ is known, given as follows.

    \begin{theorem} [Quantum amplitude estimation, {\cite[Theorem 12]{BHMT02}}] \label{thm:amplitude-estimation}
        Suppose that $U$ is a unitary operator such that 
        \begin{equation}
        U\ket{0}_{\mathsf{A}}\ket{0}_{\mathsf{B}} = \sqrt{p} \ket{0}_{\mathsf{A}}\ket{\phi_0}_{\mathsf{B}} + \sqrt{1-p} \ket{1}_{\mathsf{A}}\ket{\phi_1}_{\mathsf{B}},
        \end{equation}
        where $\ket{\phi_0}$ and $\ket{\phi_1}$ are normalized pure states, and $p \in \sbra{0, 1}$. 
        For every $\delta \in \rbra{0, 1}$, there is a quantum query algorithm that outputs $\tilde x \in \sbra{0, 1}$ with probability at least $2/3$ such that
        \begin{equation}
        \abs*{\tilde x - p} < \delta,
        \end{equation}
        using $O\rbra{1/\delta}$ queries to $U$. 
    \end{theorem}

    \cref{thm:amplitude-estimation} allows us to estimate the squared amplitude $p$ to within additive error $\delta$ with query complexity $O\rbra{1/\delta}$. 
    If one needs to estimate the square root of $p$ to within additive error $\delta'$, then one has to set $\delta = \delta'^2$ and the total query complexity is $O\rbra{1/\delta'^2}$. 
    This is because the error could become larger when taking additional arithmetic operations on the estimated value.
    To better explain this, we provide the following proposition concerning the numerical stability of square root. 

    \begin{proposition} \label{prop:sqrt-stable}
        Let $x, \tilde x \geq 0$ be two real numbers such that $\abs{x - \tilde x} < \varepsilon$ for some $\varepsilon > 0$. 
        Then, $\abs{\sqrt{x} - \sqrt{\tilde x}} < \sqrt{\varepsilon}$, and the equality holds when $x = 0$ and $\tilde x = \varepsilon$. 
    \end{proposition}
    \begin{proof}
        Without loss of generality, we assume that $\tilde x \geq x$ and let $\delta = \tilde x - x \in \rbra{0, \varepsilon}$. 
        Then, to show $\abs{\sqrt{x} - \sqrt{\tilde x}} < \sqrt{\varepsilon}$, we only need to demonstrate that for any $x \geq 0$ and $\delta > 0$, the following inequality holds:
        \begin{equation}
            \sqrt{x + \delta} - \sqrt{x} \leq \sqrt{\delta}.
        \end{equation}
        This is straightforward because we note that $\sqrt{x+\delta} \leq \sqrt{x} + \sqrt{\delta}$. If we square both sides, we obtain $x + \delta \leq x + \delta + 2\sqrt{x\delta}$, which simplifies to $\sqrt{x\delta} \geq 0$. 
        
        Therefore, it always holds that $\abs{\sqrt{x} - \sqrt{\tilde x}} < \sqrt{\varepsilon}$. In addition, it is easy to verify that the equality holds when $x = 0$ and $\tilde x = \varepsilon$, meaning that the inequality is tight.
    \end{proof}

    As demonstrated in \cref{prop:sqrt-stable}, one cannot hope to estimate $\sqrt{p}$ to within additive error $\delta'$ with query complexity better than $O\rbra{1/\delta'^2}$ by simply using the quantum amplitude estimation given in \cref{thm:amplitude-estimation}.

    \subsection{The Algorithm} \label{sec:sqrt-ampl-est-algo}

    In this subsection, we will present a quantum query algorithm for square root amplitude estimation. 
    This algorithm improves upon the quantum amplitude estimation in \cite{BHMT02}. 
    For better illustration, we need a textbook quantum algorithm for phase estimation, which was originally proposed in \cite{Kit95}. 

    \begin{theorem} [Quantum phase estimation, {\cite[Section 5.2]{NC10}}] \label{thm:PhE}
        Suppose that $U$ is a unitary operator with spectral decomposition
        \begin{equation}
            U = \sum_{j} e^{i2\pi\lambda_j} \ket{v_j},
        \end{equation}
        where $\cbra{\ket{v_j}}$ is an orthonormal basis and $\lambda_j \in [0, 1)$. 
        For every $\varepsilon \in \rbra{0, 1}$ and $\delta \in \rbra{0, 1}$, there is a quantum circuit $\mathsf{PhE}^U_{\varepsilon,\delta}$ using $O\rbra{1/\varepsilon\delta}$ queries to controlled-$U$ that performs the transform
        \begin{equation}
            \mathsf{PhE}^U_{\varepsilon,\delta} \colon \sum_{j} \alpha_j \ket{v_j} \ket{0} \mapsto \sum_{j} \alpha_j \ket{v_j} \ket{\tilde \varphi_j}
        \end{equation}
        for any coefficients $\alpha_j \in \mathbb{C}$ with $\sum_j \abs{\alpha_j}^2 = 1$, such that if we measure $\ket{\tilde \varphi_j}$ in the computational basis, then with probability at least $1-\varepsilon$ we will obtain a real number $\tilde \varphi_j \in [0, 1)$ satisfying 
        \begin{equation}
            \min\cbra*{ \abs*{\tilde \varphi_j - \lambda_j}, 1 - \abs*{\tilde \varphi_j - \lambda_j} } < \delta.
        \end{equation}
    \end{theorem}

    In the following, we present the quantum query algorithm for square root amplitude estimation. 

    \begin{theorem} [Quantum square root amplitude estimation] \label{thm:sqrt-ampl-est}
        Suppose that $U$ is a unitary operator such that 
        \begin{equation}
        U\ket{0}_{\mathsf{A}}\ket{0}_{\mathsf{B}} = \sqrt{p} \ket{0}_{\mathsf{A}}\ket{\phi_0}_{\mathsf{B}} + \sqrt{1-p} \ket{1}_{\mathsf{A}}\ket{\phi_1}_{\mathsf{B}},
        \end{equation}
        where $\ket{\phi_0}$ and $\ket{\phi_1}$ are normalized pure states, and $p \in \sbra{0, 1}$. 
        For every $\delta \in \rbra{0, 1}$, there is a quantum algorithm that outputs $\tilde x = \mathsf{SqrtAmpEst}\rbra{U, \delta} \in \sbra{0, 1}$ such that it holds with probability at least $2/3$ that
        \begin{equation}
        \abs*{\tilde x - \sqrt{p}} < \delta,
        \end{equation}
        using $O\rbra{1/\delta}$ queries to controlled-$U$ and controlled-$U^\dag$. 
    \end{theorem}
    \begin{proof}
        The idea of the proof follows that in \cite{BHMT02}. 
        Let
        \begin{equation}
        Q = -U \rbra*{I_{\mathsf{AB}} - 2 \ket{0}_{\mathsf{A}}\bra{0} \otimes \ket{0}_{\mathsf{B}}\bra{0}} U^\dag \rbra*{I_{\mathsf{AB}} - 2\ket{0}_{\mathsf{A}}\bra{0} \otimes I_{\mathsf{B}}}.
        \end{equation}
        According to the analysis of \cite[Equation (6)]{BHMT02}, 
        \begin{equation}
            U \ket{0}_{\mathsf{A}} \ket{0}_{\mathsf{B}} = - \frac{i}{\sqrt{2}} \rbra*{ e^{i\theta_p} \ket{\psi_+}_{\mathsf{AB}} - e^{-i\theta_p} \ket{\psi_-}_{\mathsf{AB}} },
        \end{equation}
        where 
        \begin{align}
            \theta_p & = \arcsin \rbra*{\sqrt{p}} \in \sbra*{0, \frac{\pi}{2}}, \\
            \ket{\psi_\pm}_{\mathsf{AB}} & = \frac{1}{\sqrt{2}} \rbra*{ \ket{0}_{\mathsf{A}}\ket{\phi_0}_{\mathsf{B}} \pm i \ket{1}_{\mathsf{A}}\ket{\phi_1}_{\mathsf{B}} }.
        \end{align}
        Note that $\ket{\psi_\pm}_{\mathsf{AB}}$ are eigenvectors of $Q$ such that $Q\ket{\psi_\pm}_{\mathsf{AB}} = e^{\pm i2\theta_p} \ket{\psi_\pm}_{\mathsf{AB}}$. 

        By \cref{thm:PhE}, if we perform $\mathsf{PhE}^Q_{\varepsilon,\delta}$ on the state $U \ket{0}_{\mathsf{A}} \ket{0}_{\mathsf{B}} \otimes \ket{0}_\mathsf{C}$, then we obtain
        \begin{equation}
        \begin{aligned}
            \mathsf{PhE}^Q_{\varepsilon,\delta} \rbra*{ U \ket{0}_{\mathsf{A}} \ket{0}_{\mathsf{B}} \otimes \ket{0}_\mathsf{C} } = - \frac{i}{\sqrt{2}} \rbra*{ e^{i\theta_p} \ket{\psi_+}_{\mathsf{AB}} \ket{\tilde \varphi_+}_{\mathsf{C}} - e^{-i\theta_p} \ket{\psi_-}_{\mathsf{AB}} \ket{\tilde \varphi_-}_{\mathsf{C}} },
        \end{aligned}
        \end{equation}
        where if we measure $\ket{\tilde \varphi_\pm}_\mathsf{C}$ in the computational basis, then with probability at least $1-\varepsilon$ we will obtain a real number $\tilde \varphi_\pm \in [0, 1)$ satisfying 
        \begin{equation} \label{eq:tilde_pm}
        \begin{aligned}
            \min\cbra*{ \abs*{\tilde \varphi_+ - \frac{\theta_p}{\pi}}, 1 - \abs*{\tilde \varphi_+ - \frac{\theta_p}{\pi}} } < \delta, \\
            \min\cbra*{ \abs*{\tilde \varphi_- - \rbra*{1-\frac{\theta_p}{\pi}}}, 1 - \abs*{\tilde \varphi_- - \rbra*{1-\frac{\theta_p}{\pi}}} } < \delta.
        \end{aligned}
        \end{equation}

        Finally, we measure $\mathsf{PhE}^Q_{\varepsilon,\delta} \rbra*{ U \ket{0}_{\mathsf{A}} \ket{0}_{\mathsf{B}} \otimes \ket{0}_\mathsf{C} }$ in the computational basis of system $\mathsf{C}$, and let $\tilde \varphi \in [0, 1)$ be the output real number. 
        Then, it can be verified that $\abs{\abs{\sin\rbra{\pi \tilde \varphi}} - \sqrt{p}} < \pi\delta$. 
        To see this, we note that with probability at least $1 - \varepsilon$, the real number $\tilde \varphi$ satisfies the condition for either $\tilde \varphi_+$ or $\tilde \varphi_-$ in \cref{eq:tilde_pm}.
        No matter which is the case, we have 
        \begin{align}
            \abs[\big]{ \abs*{\sin\rbra{\pi \tilde \varphi}} - \sqrt{p} } 
            & = \abs[\big]{\abs*{\sin\rbra{\pi \tilde \varphi}} - \abs*{\sin\rbra{\theta_p}}} \\
            & \leq \min \big\{ \abs*{ \pi \tilde \varphi - \theta_p}, \pi - \abs*{ \pi \tilde \varphi - \theta_p}, \abs*{ \pi \tilde \varphi - \rbra{\pi - \theta_p} }, \pi - \abs*{ \pi \tilde \varphi - \rbra{\pi - \theta_p}} \big\} \\
            & = \pi \min \Bigg\{ \abs*{ \tilde \varphi - \frac{\theta_p}{\pi}}, 1 - \abs*{ \tilde \varphi - \frac{\theta_p}{\pi}}, \abs*{\tilde \varphi - \rbra*{1-\frac{\theta_p}{\pi}}}, 1 - \abs*{\tilde \varphi - \rbra*{1-\frac{\theta_p}{\pi}}} \Bigg\} \\
            & < \pi \delta. 
        \end{align}

        According to the above analysis, we conclude that with probability at least $1-\varepsilon$, $\abs{\sin\rbra{\pi\tilde \varphi}}$ is $\pi\delta$-close to $\sqrt{p}$. 
        By letting $\varepsilon = 2/3$ and rescaling $\delta \gets \delta/\pi$, we obtain a quantum algorithm that estimates the value of $\sqrt{p}$ to within additive error $\delta$ with success probability at least $2/3$.
    \end{proof}

    The square root amplitude estimation given in \cref{thm:sqrt-ampl-est} essentially employs the same idea as the amplitude estimation in \cite{BHMT02}, with an extra observation that we can estimate the square root of $p$ directly by $\abs{\sin\rbra{\theta_p}}$ and note that $\abs{\sin\rbra{\cdot}}$ is $O\rbra{1}$-Lipschitz.

    To conclude this section, we explain how square root amplitude estimation can reproduce the quantum amplitude estimation in \cite{BHMT02}.
    Suppose that $U$ is a unitary operator such that $U\ket{0}_{\mathsf{A}}\ket{0}_{\mathsf{B}} = \sqrt{p} \ket{0}_{\mathsf{A}}\ket{\phi_0}_{\mathsf{B}} + \sqrt{1-p} \ket{1}_{\mathsf{A}}\ket{\phi_1}_{\mathsf{B}}$ and our goal is to estimate the value of $p$ by \cref{thm:sqrt-ampl-est} with the same (up to a constant factor) query complexity as in \cref{thm:amplitude-estimation}. 
    By \cref{thm:sqrt-ampl-est}, we can obtain (with high probability) a $\delta$-estimate $\tilde x$ of $\sqrt{p}$ with query complexity $O\rbra{1/\delta}$, i.e., $\abs{\tilde x - \sqrt{p}} < \delta$. 
    Then, we claim that $\tilde x^2$ is a $2\delta$-estimate of $p$. 
    This is seen by the fact that $\abs{\tilde x^2 - p} = \abs{\tilde x + \sqrt{p}} \cdot \abs{\tilde x - \sqrt{p}} < 2\delta$, where $0 \leq \tilde x + \sqrt{p} \leq 2$.
    By rescaling $\delta \gets \delta/2$, we can estimate the value of $p$ to within additive error $\delta$ with quantum query complexity $O\rbra{1/\delta}$. 

    \section{Pure-State Closeness Estimations} \label{sec:td-est}

    In this section, we first present a quantum query algorithm for estimating the trace distance between pure quantum states in \cref{sec:algo}. 
    Then, we present a quantum query algorithm for estimating the square root fidelity between pure quantum states in \cref{sec:algo-fi}. 
    Finally, in \cref{sec:sqr-fi-reprod}, we reproduce the pure-state squared fidelity estimation.

    \subsection{Trace Distance} \label{sec:algo}

    We present a quantum query algorithm for pure-state trace distance estimation in \cref{thm:td}, with its formal description given in \cref{algo:td}.

    \begin{algorithm*}[t]
        \caption{Quantum query algorithm for pure-state trace distance estimation.}
        \label{algo:td}
        \begin{algorithmic}[1]
        \Require Quantum oracles $U_{\varphi}$ and $U_{\psi}$ that prepare $k$-qubit pure states $\ket{\varphi}$ and $\ket{\psi}$, respectively (as well as $U_{\varphi}^\dag$, $U_{\psi}^\dag$, and their controlled versions); the desired additive error $\varepsilon \in \rbra{0, 1}$. 

        \Ensure An estimate $\tilde x$ of $\mathrm{T}\rbra{\ket{\varphi}, \ket{\psi}}$ to within additive error $\varepsilon$ with probability at least $2/3$.
        
        \State Let unitary operator $W = \rbra*{X_{\mathsf{A}} \otimes \ket{0}_{\mathsf{B}} \bra{0} + \sum\limits_{j = 1}^{2^k - 1} I_{\mathsf{A}} \otimes \ket{j}_{\mathsf{B}} \bra{j}} \cdot \rbra*{I_{\mathsf{A}} \otimes \rbra*{U_{\varphi}^\dag U_{\psi}}_{\mathsf{B}}}$.

        \State \Return $\mathsf{SqrtAmpEst}\rbra{W, \varepsilon}$ by \cref{thm:sqrt-ampl-est}.
        \end{algorithmic}
    \end{algorithm*}

    \begin{theorem} [Pure-state trace distance estimation] \label{thm:td}
        Suppose that $U_{\varphi}$ and $U_{\psi}$ are quantum unitary operators that prepare pure quantum states $\ket{\varphi} = U_{\varphi} \ket{0}$ and $\ket{\psi} = U_{\psi} \ket{0}$ (of the same dimension), respectively. 
        For $\varepsilon \in \rbra{0, 1}$, there is a quantum query algorithm that estimates $\mathrm{T}\rbra{\ket{\varphi}, \ket{\psi}}$ to within additive error $\varepsilon$ with probability at least $2/3$ using $O\rbra{1/\varepsilon}$ queries to (controlled-)$U_{\varphi}$, (controlled-)$U_{\psi}$, and their inverses. 
    \end{theorem}

    \begin{proof}
        We consider the unitary operator
        \begin{equation}
            U = U_{\varphi}^\dag U_{\psi}. 
        \end{equation}
        It can be seen that 
        \begin{equation} \label{eq:def-U}
            \bra{0} U \ket{0} = \bra{0} U_{\varphi}^\dag U_{\psi} \ket{0} = \braket{\varphi}{\psi}.
        \end{equation}
        We assume that $\ket{\varphi}$ and $\ket{\psi}$ are $k$-qubit pure quantum states. 
        That is, $\ket{\varphi}$ and $\ket{\psi}$ are $2^k$-dimensional and we write the computational basis as $\ket{0}, \ket{1}, \dots, \ket{2^k-1}$. 
        Then, by \cref{eq:def-U}, we can write
        \begin{equation} \label{eq:U-ket0}
            U \ket{0} = \braket{\varphi}{\psi} \ket{0} + \sum_{j = 1}^{2^k-1} \alpha_j \ket{j},
        \end{equation}
        where $\alpha_j \in \mathbb{C}$ for $1 \leq j < 2^k$ and it holds that 
        \begin{equation} \label{eq:sum-alpha-j}
            \abs{\braket{\varphi}{\psi}}^2 +  \sum_{j=1}^{2^k-1} \abs{\alpha_j}^2 = 1. 
        \end{equation}
        Let $\mathsf{B}$ denote the system of $k$ qubits that $U$ acts on (and we will use $U_{\mathsf{B}}$ to emphasize the system in the following analysis). 
        We introduce an auxiliary system $\mathsf{A}$ that consists of one qubit. 
        Let $V$ be a unitary operator defined by
        \begin{equation}
            V = X_{\mathsf{A}} \otimes \ket{0}_{\mathsf{B}} \bra{0} + \sum_{j = 1}^{2^k - 1} I_{\mathsf{A}} \otimes \ket{j}_{\mathsf{B}} \bra{j}, 
        \end{equation}
        where $X = \ket{0} \bra{1} + \ket{1} \bra{0}$ is the Pauli-X gate. 
        Intuitively, $V$ is used to mark those branches $\ket{j}_\mathsf{B}$ for $1 \leq j < 2^k$. 
        Let 
        \begin{equation} \label{eq:def-W}
            W = V \rbra{I_{\mathsf{A}} \otimes U_{\mathsf{B}}}.
        \end{equation}
        Then, by \cref{eq:U-ket0}, we have
        \begin{align}
            W \ket{0}_{\mathsf{A}} \ket{0}_{\mathsf{B}}
            & = V \rbra{I_{\mathsf{A}} \otimes U_{\mathsf{B}}} \ket{0}_{\mathsf{A}} \ket{0}_{\mathsf{B}} \\
            & = V \rbra*{ \braket{\varphi}{\psi} \ket{0}_{\mathsf{A}} \ket{0}_{\mathsf{B}} + \sum_{j = 1}^{2^k-1} \alpha_j \ket{0}_{\mathsf{A}} \ket{j}_{\mathsf{B}} } \\
            & = \sum_{j = 1}^{2^k-1} \alpha_j \ket{0}_{\mathsf{A}} \ket{j}_{\mathsf{B}} + \braket{\varphi}{\psi} \ket{1}_{\mathsf{A}} \ket{0}_{\mathsf{B}} \\
            & = \sqrt{p} \ket{0}_{\mathsf{A}} \ket{\phi_0}_{\mathsf{B}} + \sqrt{1 - p} \ket{1}_{\mathsf{A}} \ket{\phi_1}_{\mathsf{B}}, \label{eq:W00}
        \end{align}
        where 
        \begin{align}
            p & = \sum_{j=1}^{2^k-1} \abs{\alpha_j}^2, \label{eq:def-p} \\
            \ket{\phi_0} & = \frac{1}{\sqrt{p}} \sum_{j=1}^{2^k-1} \alpha_j \ket{j}, \label{eq:def-phi0} \\
            \ket{\phi_1} & = \frac{\braket{\varphi}{\psi}}{\sqrt{1 - p}} \ket{0}. \label{eq:def-phi1}
        \end{align}
        By \cref{thm:sqrt-ampl-est}, we can obtain an estimate 
        \begin{equation}
            \tilde x = \mathsf{SqrtAmpEst}\rbra{W, \varepsilon}
        \end{equation}
        of $\sqrt{p}$ to within additive error $\varepsilon$ (i.e., $\abs{\tilde x - \sqrt{p}} < \varepsilon$) with probability at least $2/3$ using $O\rbra{1/\varepsilon}$ queries to $W$.
        On the other hand, by \cref{eq:sum-alpha-j}, we have 
        \begin{equation} \label{eq:p-eq-td}
            \sqrt{p} = \sqrt{1 - \abs{\braket{\varphi}{\psi}}^2} 
            = \sqrt{1 - \mathrm{F}^2 \rbra{\ket{\varphi}, \ket{\psi}}}
            = \mathrm{T} \rbra{\ket{\varphi}, \ket{\psi}},
        \end{equation}
        which means that $\tilde x$ is an estimate of the trace distance $\mathrm{T} \rbra{\ket{\varphi}, \ket{\psi}}$ to within additive error $\varepsilon$ with probability at least $2/3$. 
        According to \cref{eq:def-W}, a query to $W$ consists of one query to each of $U_{\varphi}$ and $U_{\psi}$. 
        Therefore, the way we obtain $\tilde x$ uses $O\rbra{1/\varepsilon}$ queries to both $U_{\varphi}$ and $U_{\psi}$.
    \end{proof}

    \subsection{Square Root Fidelity} \label{sec:algo-fi}

    We present a quantum query algorithm for pure-state square root estimation in \cref{thm:fi}, with its formal description given in \cref{algo:fi}, which has the same structure as \cref{algo:td} for pure-state trace distance estimation.

    \begin{theorem} [Pure-state square root fidelity estimation] \label{thm:fi}
        Suppose that $U_{\varphi}$ and $U_{\psi}$ are quantum unitary operators that prepare pure quantum states $\ket{\varphi} = U_{\varphi} \ket{0}$ and $\ket{\psi} = U_{\psi} \ket{0}$ (of the same dimension), respectively. 
        For $\varepsilon \in \rbra{0, 1}$, there is a quantum query algorithm that estimates $\mathrm{F}\rbra{\ket{\varphi}, \ket{\psi}}$ to within additive error $\varepsilon$ with probability at least $2/3$ using $O\rbra{1/\varepsilon}$ queries to (controlled-)$U_{\varphi}$, (controlled-)$U_{\psi}$, and their inverses. 
    \end{theorem}
    \begin{proof}
        We also need the unitary operator $W$ defined by \cref{eq:def-W}. 
        Let 
        \begin{equation} \label{eq:def-W'}
            W' = \rbra{X_{\mathsf{A}} \otimes I_{\mathsf{B}}} W.
        \end{equation}
        By \cref{eq:W00}, we have
        \begin{equation}
            W' \ket{0}_{\mathsf{A}} \ket{0}_{\mathsf{B}} = \sqrt{p'} \ket{0}_{\mathsf{A}} \ket{\phi_1}_{\mathsf{B}} + \sqrt{1 - p'} \ket{1}_{\mathsf{A}} \ket{\phi_0}_{\mathsf{B}},
        \end{equation}
        where $p' = 1 - p$, and $p$, $\ket{\phi_0}$, and $\ket{\phi_1}$ are defined by \cref{eq:def-p}, \cref{eq:def-phi0}, and \cref{eq:def-phi1}, respectively. 
        By \cref{thm:sqrt-ampl-est}, we can obtain an estimate 
        \begin{equation}
            \tilde x = \mathsf{SqrtAmpEst}\rbra{W', \varepsilon}
        \end{equation}
        of $\sqrt{p'}$ to within additive error $\varepsilon$ (i.e., $\abs{\tilde x - \sqrt{p'}} < \varepsilon$) with probability at least $2/3$ using $O\rbra{1/\varepsilon}$ queries to $W'$.
        On the other hand, by \cref{eq:p-eq-td}, we have 
        \begin{equation}
            \sqrt{p'} = \sqrt{1 - p} = \sqrt{1 - \rbra{\mathrm{T} \rbra{\ket{\varphi}, \ket{\psi}}}^2} = \mathrm{F}\rbra{\ket{\varphi}, \ket{\psi}},
        \end{equation}
        which means that $\tilde x$ is an estimate of the square root fidelity $\mathrm{F} \rbra{\ket{\varphi}, \ket{\psi}}$ to within additive error $\varepsilon$ with probability at least $2/3$. 
        According to \cref{eq:def-W'}, a query to $W'$ consists of one query to each of $U_{\varphi}$ and $U_{\psi}$. 
        Therefore, the way we obtain $\tilde x$ uses $O\rbra{1/\varepsilon}$ queries to both $U_{\varphi}$ and $U_{\psi}$.
    \end{proof}

    \begin{algorithm*}[t]
        \caption{Quantum query algorithm for pure-state square root fidelity estimation.}
        \label{algo:fi}
        \begin{algorithmic}[1]
        \Require Quantum oracles $U_{\varphi}$ and $U_{\psi}$ that prepare $k$-qubit pure states $\ket{\varphi}$ and $\ket{\psi}$, respectively (as well as $U_{\varphi}^\dag$, $U_{\psi}^\dag$, and their controlled versions); the desired additive error $\varepsilon \in \rbra{0, 1}$. 

        \Ensure An estimate $\tilde x$ of $\mathrm{F}\rbra{\ket{\varphi}, \ket{\psi}}$ to within additive error $\varepsilon$ with probability at least $2/3$.
        
        \State Let unitary operator $W' = \rbra*{I_{\mathsf{A}} \otimes \ket{0}_{\mathsf{B}} \bra{0} + \sum\limits_{j = 1}^{2^k - 1} X_{\mathsf{A}} \otimes \ket{j}_{\mathsf{B}} \bra{j}} \cdot \rbra*{I_{\mathsf{A}} \otimes \rbra*{U_{\varphi}^\dag U_{\psi}}_{\mathsf{B}}}$.

        \State \Return $\mathsf{SqrtAmpEst}\rbra{W', \varepsilon}$ by \cref{thm:sqrt-ampl-est}.
        \end{algorithmic}
    \end{algorithm*}

    \subsection{Squared Fidelity, Revisited} \label{sec:sqr-fi-reprod}

    It is well known that the squared fidelity between two pure quantum states can be estimated to within additive error $\varepsilon$ with quantum query complexity $O\rbra{1/\varepsilon}$ based on the SWAP test \cite{BCWdW01} and the quantum amplitude estimation \cite{BHMT02}. 
    Here, we reproduce this result by using the quantum query algorithm in \cref{thm:td} as a subroutine.

    \begin{corollary} [Pure-state squared fidelity estimation, reproduced] \label{thm:sqr-fi}
        Suppose that $U_{\varphi}$ and $U_{\psi}$ are quantum unitary operators that prepare pure quantum states $\ket{\varphi} = U_{\varphi} \ket{0}$ and $\ket{\psi} = U_{\psi} \ket{0}$ (of the same dimension), respectively. 
        For $\varepsilon \in \rbra{0, 1}$, there is a quantum query algorithm that estimates $\mathrm{F}^2\rbra{\ket{\varphi}, \ket{\psi}}$ to within additive error $\varepsilon$ with probability at least $2/3$ using $O\rbra{1/\varepsilon}$ queries to (controlled-)$U_{\varphi}$, (controlled-)$U_{\psi}$, and their inverses. 
    \end{corollary}
    \begin{proof}
        By \cref{eq:relation-td-fi}, we have 
        \begin{equation}
        \mathrm{F}^2\rbra{\ket{\varphi}, \ket{\psi}} = 1 - \rbra*{\mathrm{T}\rbra{\ket{\varphi}, \ket{\psi}}}^2.
        \end{equation}
        By the quantum query algorithm for pure-state trace distance estimation given in \cref{thm:td} using $O\rbra{1/\varepsilon}$ queries to $U_{\varphi}$ and $U_{\psi}$, we can, with probability at least $2/3$, obtain an estimate $\tilde x \in \sbra{0, 1}$ of $\mathrm{T}\rbra{\ket{\varphi}, \ket{\psi}}$ to within additive error $\varepsilon/2$, i.e.,
        \begin{equation} \label{eq:sqr-fi-pf-td}
            \abs*{\tilde x - \mathrm{T}\rbra{\ket{\varphi}, \ket{\psi}}} < \frac{\varepsilon}{2}.
        \end{equation}
        Then, we can choose $1 - \tilde x^2$ to be the estimate of $\mathrm{F}^2\rbra{\ket{\varphi}, \ket{\psi}}$ to within additive error $\varepsilon$. 
        To see this, we note that by \cref{eq:sqr-fi-pf-td}, we have
        \begin{align}
            \abs*{\rbra{1 - \tilde x^2} - \mathrm{F}^2\rbra{\ket{\varphi}, \ket{\psi}}} 
            & = \abs*{\tilde x^2 - \rbra*{\mathrm{T}\rbra{\ket{\varphi}, \ket{\psi}}}^2} \\
            & = \abs[\big]{\tilde x + \mathrm{T}\rbra{\ket{\varphi}, \ket{\psi}}} \cdot \abs[\big]{\tilde x - \mathrm{T}\rbra{\ket{\varphi}, \ket{\psi}}} \\
            & < 2 \cdot \frac \varepsilon 2 = \varepsilon.
        \end{align}
    \end{proof}

    \section{Query Lower Bounds} \label{sec:quantum-lower-bounds}

    In this section, we prove quantum query lower bounds for estimating the trace distance, squared fidelity and square root fidelity between pure quantum states in \cref{sec:lb-td}, \cref{sec:lb-sqr-fi} and \cref{sec:lb-fi}, respectively. 
    To prove the lower bounds, we need the quantum query lower bound for distinguishing probability distributions given in \cite{Bel19}, which is obtained by the quantum adversary method \cite{Amb02}. 

    \begin{theorem} [Quantum query lower bound for distinguishing probability distributions, {\cite[Theorem 4]{Bel19}}] \label{thm:qlower-prob-distri}
        Let $p, q \colon \sbra{n} \to \sbra{0, 1}$ be two probability distributions on a sample space of $n$ elements.
        Suppose that $U_p$ and $U_q$ are quantum unitary oracles such that
        \begin{align}
            U_p \ket{0} & = \sum_{j \in \sbra{n}} \sqrt{p\rbra{j}} \ket{j}, \\
            U_q \ket{0} & = \sum_{j \in \sbra{n}} \sqrt{q\rbra{j}} \ket{j}.
        \end{align}
        Then, any quantum query algorithm using queries to quantum oracle $U$ (as well as $U^\dag$ and their controlled versions) that determines whether $U = U_p$ or $U = U_q$ with probability at least $2/3$, promised that it is in either case, has query complexity $\Omega\rbra{1/d_{\textup{H}}\rbra{p, q}}$, where 
        \begin{equation}
            d_{\textup{H}}\rbra{p, q} = \sqrt{\frac 1 2 \sum_{j \in \sbra{n}} \rbra*{\sqrt{p\rbra{j}} - \sqrt{q\rbra{j}}}^2}
        \end{equation}
        is the Hellinger distance.
    \end{theorem}

    \subsection{Trace Distance} \label{sec:lb-td}

    We prove a matching quantum query lower bound for pure-state trace distance estimation in \cref{thm:qlower-td}. 
    This is done by a reduction from distinguishing certain probability distributions to discriminating certain pure quantum states, where the latter can be solved by pure-state trace distance estimation.

    \begin{theorem} [Query lower bounds for pure-state trace distance estimation] \label{thm:qlower-td}
        For $\varepsilon \in \rbra{0, 1/2}$, any quantum query algorithm that estimates the trace distance between two pure quantum states to within additive error $\varepsilon$ with probability at least $2/3$ has query complexity $\Omega\rbra{1/\varepsilon}$. 
    \end{theorem}
    \begin{proof}
        We consider two probability distributions $p^+$ and $p^-$ on a sample space of $n$ elements where $n$ is even, and they are defined by
        \begin{equation}
            p^\pm \rbra{j} = \frac{1 \pm \rbra{-1}^j 2 \varepsilon}{n}
        \end{equation}
        for every $j \in \sbra{n}$. 
        Note that the Hellinger distance between $p^+$ and $p^-$ is
        \begin{equation}
            d_{\textup{H}}\rbra{p^+, p^-} = \sqrt{1 - \sqrt{1 - 4\varepsilon^2}} \leq 2\varepsilon. 
        \end{equation}
        By \cref{thm:qlower-prob-distri}, any quantum query algorithm using queries to quantum oracle $U$ that determines whether $U = U_{p^+}$ or $U = U_{p^-}$ with probability at least $2/3$, promised that it is in either case, has query complexity $\Omega\rbra{1/d_{\textup{H}}\rbra{p^+, p^-}} = \Omega\rbra{1/\varepsilon}$, 
        where
        \begin{equation}
            U_{p^\pm} \ket{0} = \sum_{j \in \sbra{n}} \sqrt{p^\pm\rbra{j}} \ket{j}.
        \end{equation}
        On the other hand, suppose that there is a quantum query algorithm $\mathcal{A}$ that, with probability at least $2/3$, estimates to within additive error $\varepsilon$ the trace distance between two pure quantum states, given quantum oracles that prepare them, with query complexity $Q$. 
        Now we will construct another quantum query algorithm $\mathcal{A}'$ with $\mathcal{A}$ as a subroutine, which determines whether the quantum oracle $U$ satisfies $U = U_{p^+}$ or $U = U_{p^-}$, promised that it is in either case. 
        The algorithm $\mathcal{A}'$ consists of two steps given as follows.
        \begin{enumerate}
            \item Let $d$ be the estimate of the trace distance between $U\ket{0}$ and $U_{p^+}\ket{0}$ to within additive error $\varepsilon$ by applying $\mathcal{A}$ with queries to $U$ and $U_{p^+}$. 
            Note that this step uses $Q$ queries to $U$. 
            \item If $d < \varepsilon$, then return that $U = U_{p^+}$; and return that $U = U_{p^-}$ otherwise. 
        \end{enumerate}
        It is clear that the query complexity of $\mathcal{A}'$ is $Q$. 
        To see the correctness of the algorithm $\mathcal{A}'$, we consider two cases $U = U_{p^+}$ and $U = U_{p^-}$ separately as follows. 
        \begin{itemize}
            \item $U = U_{p^+}$. In this case, the trace distance between $U\ket{0}$ and $U_{p^+}\ket{0}$ is $0$. Therefore, it holds that $d < \varepsilon$ with probability at least $2/3$. 
            According to Step 2 of $\mathcal{A}'$, we conclude that it will return $U = U_{p^+}$ correctly with probability at least $2/3$. 
            \item $U = U_{p^-}$. In this case, the trace distance between $U\ket{0}$ and $U_{p^+}\ket{0}$ is
            \begin{align}
                \mathrm{T}\rbra{U_{p^+}\ket{0}, U_{p^-}\ket{0}} 
                & = \sqrt{1 - \mathrm{F}^2\rbra{U_{p^+}\ket{0}, U_{p^-}\ket{0}}} \\
                & = \sqrt{1 - \rbra{1 - 4\varepsilon^2}} = 2\varepsilon.
            \end{align}
            Therefore, it holds that $d > \varepsilon$ with probability at least $2/3$. 
            According to Step 2 of $\mathcal{A}'$, we conclude that it will return $U = U_{p^-}$ correctly with probability at least $2/3$. 
        \end{itemize}
        Now we have shown that $\mathcal{A}'$ is a quantum query algorithm with query complexity $Q$ that determines whether $U = U_{p^+}$ or $U = U_{p^-}$ with probability at least $2/3$, promised that it is in either case. 
        This directly gives that $Q = \Omega\rbra{1/d_{\textup{H}}\rbra{p^+, p^-}} = \Omega\rbra{1/\varepsilon}$, and thus yields the proof.
    \end{proof}

    \subsection{Squared Fidelity} \label{sec:lb-sqr-fi}

    The quantum query lower bound for pure-state squared fidelity estimation is implied from the quantum query lower bound for quantum counting \cite{BBC+01,NW99}.
    Here, we give a different proof by a reduction from the quantum query lower bound for distinguishing probability distributions, which is inspired by the proof of \cref{thm:qlower-td}. 

    \begin{theorem} [Query lower bounds for pure-state squared fidelity estimation] \label{thm:qlower-sqr-fi}
        For $\varepsilon \in \rbra{0, 1/2}$, any quantum query algorithm that estimates the square root fidelity between two pure quantum states to within additive error $\varepsilon$ with probability at least $2/3$ has query complexity $\Omega\rbra{1/\varepsilon}$. 
    \end{theorem}
    \begin{proof}
        The proof uses the same choice of the probability distributions $p^+$ and $p^-$ as that of \cref{thm:qlower-td}.
        In the following, we will use the same notations as in the proof of \cref{thm:qlower-td}.
        Suppose that there is a quantum query algorithm $\mathcal{A}$ that, with probability at least $2/3$, estimates to within additive error $\varepsilon$ the squared fidelity between two pure quantum states, given quantum oracles that prepare them, with query complexity $Q$. 
        Now we will construct another quantum query algorithm $\mathcal{A}'$ with $\mathcal{A}$ as a subroutine, which determines whether the quantum oracle $U$ satisfies $U = U_{p^+}$ or $U = U_{p^-}$, promised that it is in either case. 
        The algorithm $\mathcal{A}'$ consists of two steps given as follows.
        \begin{enumerate}
            \item Let $d$ be the estimate of the squared fidelity between $U\ket{0}$ and $\ket{\phi}$ to within additive error $\varepsilon$ (with success probability at least $8/9$) by applying $\mathcal{A}$ with $O\rbra{Q}$ queries to $U$, where
            \begin{equation}
                \ket{\phi} = \sqrt{\frac{2}{n}} \sum_{j \in \sbra{n/2}} \ket{2j}.
            \end{equation}
            \item If $d > 1/2$, then return that $U = U_{p^+}$; and return that $U = U_{p^-}$ otherwise. 
        \end{enumerate}
        It is clear that the query complexity of $\mathcal{A}'$ is $O\rbra{Q}$. 
        To see the correctness of the algorithm $\mathcal{A}'$, we consider two cases $U = U_{p^+}$ and $U = U_{p^-}$ separately as follows. 
        \begin{itemize}
            \item $U = U_{p^+}$. In this case, the squared fidelity between $U\ket{0}$ and $\ket{\phi}$ is
            \begin{equation}
                \mathrm{F}^2\rbra{U_{p^+}\ket{0},\ket{\phi}} = \frac 1 2 + \varepsilon.
            \end{equation}
            Therefore, it holds that $d > 1/2$ with probability at least $2/3$. 
            According to Step 2 of $\mathcal{A}'$, we conclude that it will return $U = U_{p^+}$ correctly with probability at least $2/3$. 
            \item $U = U_{p^-}$. In this case, the squared fidelity between $U\ket{0}$ and $\ket{\phi}$ is
            \begin{equation}
                \mathrm{F}^2\rbra{U_{p^-}\ket{0},\ket{\phi}} = \frac 1 2 - \varepsilon.
            \end{equation}
            Therefore, it holds that $d < 1/2$ with probability at least $2/3$. 
            According to Step 2 of $\mathcal{A}'$, we conclude that it will return $U = U_{p^-}$ correctly with probability at least $2/3$. 
        \end{itemize}
        Now we have shown that $\mathcal{A}'$ is a quantum query algorithm with query complexity $Q$ that determines whether $U = U_{p^+}$ or $U = U_{p^-}$ with probability at least $2/3$, promised that it is in either case. 
        This directly gives that $Q = \Omega\rbra{1/d_{\textup{H}}\rbra{p^+, p^-}} = \Omega\rbra{1/\varepsilon}$, and thus yields the proof.
    \end{proof}

    \subsection{Square Root Fidelity} \label{sec:lb-fi}

    As a corollary of \cref{thm:qlower-sqr-fi}, we also give a quantum query lower bound for pure-state square root fidelity estimation. 

    \begin{theorem} [Query lower bounds for pure-state square root fidelity estimation] \label{thm:qlb-sqrt-fi}
        For $\varepsilon \in \rbra{0, 1/4}$, any quantum query algorithm that estimates the square root fidelity between two pure quantum states to within additive error $\varepsilon$ with probability at least $2/3$ has query complexity $\Omega\rbra{1/\varepsilon}$. 
    \end{theorem}
    \begin{proof}
        Suppose that there is a quantum query algorithm for pure-state square root fidelity estimation with query complexity $Q\rbra{\varepsilon}$ to within additive error $\varepsilon$ with success probability at least $2/3$. 
        Let $\tilde x$ be the estimate (produced by this algorithm with query complexity $Q\rbra{\varepsilon/2}$) of the square root fidelity between two pure quantum states $\ket{\varphi}$ and $\ket{\psi}$ to within additive error $\varepsilon/2$.
        Then, $\tilde x^2$ is an estimate of $\mathrm{F}^2\rbra{\ket{\varphi},\ket{\psi}}$ to within additive error $\varepsilon$. 
        To see this, we note that
        \begin{align}
            \abs*{\tilde x^2 - \mathrm{F}^2\rbra{\ket{\varphi},\ket{\psi}}}
            & = \abs[\big]{\tilde x + \mathrm{F}\rbra{\ket{\varphi},\ket{\psi}}} \cdot \abs[\big]{\tilde x - \mathrm{F}\rbra{\ket{\varphi},\ket{\psi}}} \\
            & \leq 2 \cdot \frac{\varepsilon}{2} = \varepsilon.
        \end{align}
        By \cref{thm:qlower-sqr-fi}, it holds that 
        \begin{equation}
            Q\rbra{\varepsilon/2} = \Omega\rbra{1/\varepsilon},
        \end{equation}
        which implies that $Q\rbra{\varepsilon} = \Omega\rbra{1/\varepsilon}$. 
    \end{proof}

    \section*{Acknowledgment}

    The author would like to thank Fran\c{c}ois Le Gall for pointing out the related work \cite{KLLP19}, Minbo Gao for pointing out the related lecture notes \cite{dW19}, and the anonymous reviewers for their constructive suggestions.

    This work was supported by the Ministry of Education, Culture, Sports, Science and Technology (MEXT) Quantum Leap Flagship Program (MEXT Q-LEAP) under Grant JPMXS0120319794.

    \addcontentsline{toc}{section}{References}

    \bibliographystyle{unsrturl}
    \bibliography{main}

\begin{thebibliography}{10}

\bibitem{NC10}
Michael~A. Nielsen and Isaac~L. Chuang.
\newblock {\em Quantum Computation and Quantum Information}.
\newblock Cambridge University Press, 2010.
\newblock \href {https://doi.org/10.1017/CBO9780511976667} {\path{doi:10.1017/CBO9780511976667}}.

\bibitem{Uhl76}
Armin Uhlmann.
\newblock The “transition probability” in the state space of a *-algebra.
\newblock {\em Reports on Mathematical Physics}, 9(2):273--279, 1976.
\newblock \href {https://doi.org/10.1016/0034-4877(76)90060-4} {\path{doi:10.1016/0034-4877(76)90060-4}}.

\bibitem{Joz94}
Richard Jozsa.
\newblock Fidelity for mixed quantum states.
\newblock {\em Journal of Modern Optics}, 41(12):2315--2323, 1994.
\newblock \href {https://doi.org/10.1080/09500349414552171} {\path{doi:10.1080/09500349414552171}}.

\bibitem{FvdG99}
Christopher~A. Fuchs and Jeroen van~de Graaf.
\newblock Cryptographic distinguishability measures for quantum-mechanical states.
\newblock {\em IEEE Transactions on Information Theory}, 45(4):1216--1227, 1999.
\newblock \href {https://doi.org/10.1109/18.761271} {\path{doi:10.1109/18.761271}}.

\bibitem{Wil13}
Mark~M. Wilde.
\newblock {\em Quantum Information Theory}.
\newblock Cambridge University Press, 2013.
\newblock \href {https://doi.org/10.1017/9781316809976} {\path{doi:10.1017/9781316809976}}.

\bibitem{AFCB14}
Mateus Ara{\'{u}}jo, Adrien Feix, Fabio Costa, and {\v{C}}aslav Brukner.
\newblock Quantum circuits cannot control unknown operations.
\newblock {\em New Journal of Physics}, 16(9):093026, 2014.
\newblock \href {https://doi.org/10.1088/1367-2630/16/9/093026} {\path{doi:10.1088/1367-2630/16/9/093026}}.

\bibitem{QDS+19}
Marco~T{\'{u}}lio Quintino, Qingxiuxiong Dong, Atsushi Shimbo, Akihito Soeda, and Mio Murao.
\newblock Reversing unknown quantum transformations: universal quantum circuit for inverting general unitary operations.
\newblock {\em Physical Review Letters}, 123(21):210502, 2019.
\newblock \href {https://doi.org/10.1103/PhysRevLett.123.210502} {\path{doi:10.1103/PhysRevLett.123.210502}}.

\bibitem{GL20}
Andr{\'a}s Gily{\'e}n and Tongyang Li.
\newblock Distributional property testing in a quantum world.
\newblock In {\em Proceedings of the 11th Innovations in Theoretical Computer Science Conference}, pages 25:1--25:19, 2020.
\newblock \href {https://doi.org/10.4230/LIPIcs.ITCS.2020.25} {\path{doi:10.4230/LIPIcs.ITCS.2020.25}}.

\bibitem{vACGN23}
Joran van Apeldoorn, Arjan Cornelissen, Andr{\'{a}}s Gily{\'{e}}n, and Giacomo Nannicini.
\newblock Quantum tomography using state-preparation unitaries.
\newblock In {\em Proceedings of the 2023 Annual ACM-SIAM Symposium on Discrete Algorithms}, pages 1265--1318, 2023.
\newblock \href {https://doi.org/10.1137/1.9781611977554.ch47} {\path{doi:10.1137/1.9781611977554.ch47}}.

\bibitem{Wat02}
John Watrous.
\newblock Limits on the power of quantum statistical zero-knowledge.
\newblock In {\em Proceedings of the 43rd Annual IEEE Symposium on Foundations of Computer Science}, pages 459--468, 2002.
\newblock \href {https://doi.org/10.1109/SFCS.2002.1181970} {\path{doi:10.1109/SFCS.2002.1181970}}.

\bibitem{HHJ+17}
Jeongwan Haah, Aram~W. Harrow, Zhengfeng Ji, Xiaodi Wu, and Nengkun Yu.
\newblock Sample-optimal tomography of quantum states.
\newblock {\em IEEE Transactions on Information Theory}, 63(9):5628--5641, 2017.
\newblock \href {https://doi.org/10.1109/TIT.2017.2719044} {\path{doi:10.1109/TIT.2017.2719044}}.

\bibitem{OW16}
Ryan O'Donnell and John Wright.
\newblock Efficient quantum tomography.
\newblock In {\em Proceedings of the 48th Annual ACM Symposium on Theory of Computing}, pages 899--912, 2016.
\newblock \href {https://doi.org/10.1145/2897518.2897544} {\path{doi:10.1145/2897518.2897544}}.

\bibitem{BOW19}
Costin B{\u{a}}descu, Ryan O'Donnell, and John Wright.
\newblock Quantum state certification.
\newblock In {\em Proceedings of the 51st Annual ACM SIGACT Symposium on Theory of Computing}, pages 503--514, 2019.
\newblock \href {https://doi.org/10.1145/3313276.3316344} {\path{doi:10.1145/3313276.3316344}}.

\bibitem{BCWdW01}
Harry Buhrman, Richard Cleve, John Watrous, and Ronald de~Wolf.
\newblock Quantum fingerprinting.
\newblock {\em Physical Review Letters}, 87(16):167902, 2001.
\newblock \href {https://doi.org/10.1103/PhysRevLett.87.167902} {\path{doi:10.1103/PhysRevLett.87.167902}}.

\bibitem{FL11}
Steven~T. Flammia and Yi-Kai Liu.
\newblock Direct fidelity estimation from few pauli measurements.
\newblock {\em Physical Review Letters}, 106(23):230501, 2011.
\newblock \href {https://doi.org/10.1103/PhysRevLett.106.230501} {\path{doi:10.1103/PhysRevLett.106.230501}}.

\bibitem{GT09}
Otfried G\"{u}hne and G\'{e}za T\'{o}th.
\newblock Entanglement detection.
\newblock {\em Physics Reports}, 474(1-6):1--75, 2009.
\newblock \href {https://doi.org/10.1016/j.physrep.2009.02.004} {\path{doi:10.1016/j.physrep.2009.02.004}}.

\bibitem{ALL22}
Anurag Anshu, Zeph Landau, and Yunchao Liu.
\newblock Distributed quantum inner product estimation.
\newblock In {\em Proceedings of the 54th Annual ACM SIGACT Symposium on Theory of Computing}, pages 44--51, 2022.
\newblock \href {https://doi.org/10.1145/3519935.3519974} {\path{doi:10.1145/3519935.3519974}}.

\bibitem{WZC+23}
Qisheng Wang, Zhicheng Zhang, Kean Chen, Ji~Guan, Wang Fang, Junyi Liu, and Mingsheng Ying.
\newblock Quantum algorithm for fidelity estimation.
\newblock {\em IEEE Transactions on Information Theory}, 69(1):273--282, 2023.
\newblock \href {https://doi.org/10.1109/TIT.2022.3203985} {\path{doi:10.1109/TIT.2022.3203985}}.

\bibitem{WGL+22}
Qisheng Wang, Ji~Guan, Junyi Liu, Zhicheng Zhang, and Mingsheng Ying.
\newblock New quantum algorithms for computing quantum entropies and distances.
\newblock {\em IEEE Transactions on Information Theory}, 70(8):5653--5680, 2024.
\newblock \href {https://doi.org/10.1109/TIT.2024.3399014} {\path{doi:10.1109/TIT.2024.3399014}}.

\bibitem{GP22}
Andr\'{a}s {Gily\'{e}n} and Alexander Poremba.
\newblock Improved quantum algorithms for fidelity estimation.
\newblock ArXiv e-prints, 2022.
\newblock \href {https://arxiv.org/abs/2203.15993} {\path{arXiv:2203.15993}}.

\bibitem{LWWZ24}
Nana Liu, Qisheng Wang, Mark~M. Wilde, and Zhicheng Zhang.
\newblock Quantum algorithms for matrix geometric means.
\newblock ArXiv e-prints, 2024.
\newblock \href {https://arxiv.org/abs/2405.00673} {\path{arXiv:2405.00673}}.

\bibitem{WZ23}
Qisheng Wang and Zhicheng Zhang.
\newblock Fast quantum algorithms for trace distance estimation.
\newblock {\em IEEE Transactions on Information Theory}, 70(4):2720--2733, 2024.
\newblock \href {https://doi.org/10.1109/TIT.2023.3321121} {\path{doi:10.1109/TIT.2023.3321121}}.

\bibitem{ZRC19}
Jiaju Zhang, Paola Ruggiero, and Pasquale Calabrese.
\newblock Subsystem trace distance in quantum field theory.
\newblock {\em Physical Review Letters}, 122(14):141602, 2019.
\newblock \href {https://doi.org/10.1103/PhysRevLett.122.141602} {\path{doi:10.1103/PhysRevLett.122.141602}}.

\bibitem{ZR23}
Jiaju Zhang and M.~A. Rajabpour.
\newblock Trace distance between fermionic {Gaussian} states from a truncation method.
\newblock {\em Physical Review A}, 108(2):022414, 2023.
\newblock \href {https://doi.org/10.1103/PhysRevA.108.022414} {\path{doi:10.1103/PhysRevA.108.022414}}.

\bibitem{CPCC20}
Marco Cerezo, Alexander Poremba, Lukasz Cincio, and Patrick~J. Coles.
\newblock Variational quantum fidelity estimation.
\newblock {\em Quantum}, 4:248, 2020.
\newblock \href {https://doi.org/10.22331/q-2020-03-26-248} {\path{doi:10.22331/q-2020-03-26-248}}.

\bibitem{CSZW22}
Ranyiliu Chen, Zhixin Song, Xuanqiang Zhao, and Xin Wang.
\newblock Variational quantum algorithms for trace distance and fidelity estimation.
\newblock {\em Quantum Science and Technology}, 7(1):015019, 2022.
\newblock \href {https://doi.org/10.1088/2058-9565/ac38ba} {\path{doi:10.1088/2058-9565/ac38ba}}.

\bibitem{TV21}
Kok~Chuan Tan and Tyler Volkoff.
\newblock Variational quantum algorithms to estimate rank, quantum entropies, fidelity, and fisher information via purity minimization.
\newblock {\em Physical Review Research}, 3(3):033251, 2021.
\newblock \href {https://doi.org/10.1103/PhysRevResearch.3.033251} {\path{doi:10.1103/PhysRevResearch.3.033251}}.

\bibitem{Wat09}
John Watrous.
\newblock Zero-knowledge against quantum attacks.
\newblock {\em SIAM Journal on Computing}, 39(1):25--58, 2009.
\newblock \href {https://doi.org/10.1137/060670997} {\path{doi:10.1137/060670997}}.

\bibitem{RASW23}
Soorya Rethinasamy, Rochisha Agarwal, Kunal Sharma, and Mark~M. Wilde.
\newblock Estimating distinguishability measures on quantum computers.
\newblock {\em Physical Review A}, 108(1):012409, 2023.
\newblock \href {https://doi.org/10.1103/PhysRevA.108.012409} {\path{doi:10.1103/PhysRevA.108.012409}}.

\bibitem{LGLW23}
Fran{\c{c}}ois Le~Gall, Yupan Liu, and Qisheng Wang.
\newblock Space-bounded quantum state testing via space-efficient quantum singular value transformation.
\newblock ArXiv e-prints, 2023.
\newblock \href {https://arxiv.org/abs/2308.05079} {\path{arXiv:2308.05079}}.

\bibitem{VV17}
Gregory Valiant and Paul Valiant.
\newblock Estimating the unseen: improved estimators for entropy and other properties.
\newblock {\em Journal of the ACM}, 64(6):37:1--37:41, 2017.
\newblock \href {https://doi.org/10.1145/3125643} {\path{doi:10.1145/3125643}}.

\bibitem{KLLP19}
Iordanis Kerenidis, Jonas Landman, Alessandro Luongo, and Anupam Prakash.
\newblock $q$-means: a quantum algorithm for unsupervised machine learning.
\newblock In {\em Advances in Neural Information Processing Systems}, volume~32, pages 4134--4144, 2019.
\newblock URL: \url{https://proceedings.neurips.cc/paper/2019/hash/16026d60ff9b54410b3435b403afd226-Abstract.html}.

\bibitem{BHMT02}
Gilles Brassard, Peter H{\o}yer, Michele Mosca, and Alain Tapp.
\newblock Quantum amplitude amplification and estimation.
\newblock In Samuel~J. Lomonaco, Jr. and Howard~E. Brandt, editors, {\em Quantum Computation and Information}, volume 305 of {\em Contemporary Mathematics}, pages 53--74. AMS, 2002.
\newblock \href {https://doi.org/10.1090/conm/305/05215} {\path{doi:10.1090/conm/305/05215}}.

\bibitem{BBC+01}
Robert Beals, Harry Buhrman, Richard Cleve, Michele Mosca, and Ronald de~Wolf.
\newblock Quantum lower bounds by polynomials.
\newblock {\em Journal of the ACM}, 48(4):778--797, 2001.
\newblock \href {https://doi.org/10.1145/502090.502097} {\path{doi:10.1145/502090.502097}}.

\bibitem{NW99}
Ashwin Nayak and Felix Wu.
\newblock The quantum query complexity of approximating the median and related statistics.
\newblock In {\em Proceedings of the 31st Annual ACM Symposium on Theory of Computing}, pages 384--393, 1999.
\newblock \href {https://doi.org/10.1145/301250.301349} {\path{doi:10.1145/301250.301349}}.

\bibitem{ZLY22}
Xiao-Ming Zhang, Tongyang Li, and Xiao Yuan.
\newblock Quantum state preparation with optimal circuit depth: implementations and applications.
\newblock {\em Physical Review Letters}, 129(23):230504, 2022.
\newblock \href {https://doi.org/10.1103/PhysRevLett.129.230504} {\path{doi:10.1103/PhysRevLett.129.230504}}.

\bibitem{STY+23}
Xiaoming Sun, Guojing Tian, Shuai Yang, Pei Yuan, and Shengyu Zhang.
\newblock Asymptotically optimal circuit depth for quantum state preparation and general unitary synthesis.
\newblock {\em IEEE Transactions on Computer-Aided Design of Integrated Circuits and Systems}, 42(10):3301--3314, 2023.
\newblock \href {https://doi.org/10.1109/TCAD.2023.3244885} {\path{doi:10.1109/TCAD.2023.3244885}}.

\bibitem{Ros21}
Gregory Rosenthal.
\newblock Query and depth upper bounds for quantum unitaries via {Grover} search.
\newblock ArXiv e-prints, 2021.
\newblock \href {https://arxiv.org/abs/2111.07992} {\path{arXiv:2111.07992}}.

\bibitem{YZ23}
Pei Yuan and Shengyu Zhang.
\newblock Optimal (controlled) quantum state preparation and improved unitary synthesis by quantum circuits with any number of ancillary qubits.
\newblock {\em Quantum}, 7:956, 2023.
\newblock \href {https://doi.org/10.22331/q-2023-03-20-956} {\path{doi:10.22331/q-2023-03-20-956}}.

\bibitem{GLM08}
Vittorio Giovannetti, Seth Lloyd, and Lorenzo Maccone.
\newblock Quantum random access memory.
\newblock {\em Physical Review Letters}, 100(16):160501, 2008.
\newblock \href {https://doi.org/10.1103/PhysRevLett.100.160501} {\path{doi:10.1103/PhysRevLett.100.160501}}.

\bibitem{KP17}
Iordanis Kerenidis and Anupam Prakash.
\newblock Quantum recommendation systems.
\newblock In {\em Proceedings of the 8th Innovations in Theoretical Computer Science Conference}, pages 49:1--49:21, 2017.
\newblock \href {https://doi.org/10.4230/LIPIcs.ITCS.2017.49} {\path{doi:10.4230/LIPIcs.ITCS.2017.49}}.

\bibitem{Che00}
Anthony Chefles.
\newblock Quantum state discrimination.
\newblock {\em Contemporary Physics}, 41(6):401--424, 2000.
\newblock \href {https://doi.org/10.1080/00107510010002599} {\path{doi:10.1080/00107510010002599}}.

\bibitem{BC09}
Stephen~M. Barnett and Sarah Croke.
\newblock Quantum state discrimination.
\newblock {\em Advances in Optics and Photonics}, 1(2):238--278, 2009.
\newblock \href {https://doi.org/10.1364/AOP.1.000238} {\path{doi:10.1364/AOP.1.000238}}.

\bibitem{BK15}
Joonwoo Bae and Leong-Chuan Kwek.
\newblock Quantum state discrimination and its applications.
\newblock {\em Journal of Physics A: Mathematical and Theoretical}, 48(8):083001, 2015.
\newblock \href {https://doi.org/10.1088/1751-8113/48/8/083001} {\path{doi:10.1088/1751-8113/48/8/083001}}.

\bibitem{Hel67}
Carl~W. Helstrom.
\newblock Detection theory and quantum mechanics.
\newblock {\em Information and Control}, 10(3):254--291, 1967.
\newblock \href {https://doi.org/10.1016/S0019-9958(67)90302-6} {\path{doi:10.1016/S0019-9958(67)90302-6}}.

\bibitem{Hol73}
Alexander~S. Holevo.
\newblock Statistical decision theory for quantum systems.
\newblock {\em Journal of Multivariate Analysis}, 3(4):337--394, 1973.
\newblock \href {https://doi.org/10.1016/0047-259X(73)90028-6} {\path{doi:10.1016/0047-259X(73)90028-6}}.

\bibitem{Wot81}
W.~K. Wootters.
\newblock Statistical distance and {Hilbert} space.
\newblock {\em Physical Review D}, 23(2):357, 1981.
\newblock \href {https://doi.org/10.1103/PhysRevD.23.357} {\path{doi:10.1103/PhysRevD.23.357}}.

\bibitem{dW19}
Ronald de~Wolf.
\newblock Quantum computing: lecture notes.
\newblock ArXiv e-prints, 2019.
\newblock \href {https://arxiv.org/abs/1907.09415v5} {\path{arXiv:1907.09415v5}}.

\bibitem{Kit95}
A.~Yu. Kitaev.
\newblock Quantum measurements and the {Abelian} stabilizer problem.
\newblock ArXiv e-prints, 1995.
\newblock \href {https://arxiv.org/abs/quant-ph/9511026} {\path{arXiv:quant-ph/9511026}}.

\bibitem{Bel19}
Aleksandrs Belovs.
\newblock Quantum algorithms for classical probability distributions.
\newblock In {\em Proceedings of the 27th Annual European Symposium on Algorithms}, pages 16:1--16:11, 2019.
\newblock \href {https://doi.org/10.4230/LIPIcs.ESA.2019.16} {\path{doi:10.4230/LIPIcs.ESA.2019.16}}.

\bibitem{Pan08}
Liam Paninski.
\newblock A coincidence-based test for uniformity given very sparsely sampled discrete data.
\newblock {\em IEEE Transactions on Information Theory}, 54(10):4750--4755, 2008.
\newblock \href {https://doi.org/10.1109/TIT.2008.928987} {\path{doi:10.1109/TIT.2008.928987}}.

\bibitem{LWL23}
Jingquan Luo, Qisheng Wang, and Lvzhou Li.
\newblock Succinct quantum testers for closeness and $k$-wise uniformity of probability distributions.
\newblock {\em IEEE Transactions on Information Theory}, 70(7):5092--5103, 2024.
\newblock \href {https://doi.org/10.1109/TIT.2024.3393756} {\path{doi:10.1109/TIT.2024.3393756}}.

\bibitem{OW21}
Ryan O'Donnell and John Wright.
\newblock Quantum spectrum testing.
\newblock {\em Communications in Mathematical Physics}, 387(1):1--75, 2021.
\newblock \href {https://doi.org/10.1007/s00220-021-04180-1} {\path{doi:10.1007/s00220-021-04180-1}}.

\bibitem{BHH11}
Sergey Bravyi, Aram~W. Harrow, and Avinatan Hassidim.
\newblock Quantum algorithms for testing properties of distributions.
\newblock {\em IEEE Transactions on Information Theory}, 57(6):3971--3981, 2011.
\newblock \href {https://doi.org/10.1109/TIT.2011.2134250} {\path{doi:10.1109/TIT.2011.2134250}}.

\bibitem{Amb02}
Andris Ambainis.
\newblock Quantum lower bounds by quantum arguments.
\newblock {\em Journal of Computer and System Sciences}, 64(4):750--767, 2002.
\newblock \href {https://doi.org/10.1006/jcss.2002.1826} {\path{doi:10.1006/jcss.2002.1826}}.

\end{thebibliography}

    \appendix

    \section{Folklore Approaches} \label[appendix]{sec:folklore}

    The folklore approaches for estimating the trace distance and square root fidelity are based on the SWAP test \cite{BCWdW01}. 
    The idea is simple: first estimate the squared fidelity between two pure quantum states, and then compute their trace distance and square root fidelity by the estimate of their squared fidelity. 
    We include these folklore approaches here for reference.
    First, the sample complexity and query complexity for pure-state squared fidelity estimation is given in \cref{sec:folklore-sqr-fi}.
    Then, using these algorithms, we derive the sample complexity and query complexity for pure-state square root fidelity and trace distance estimations in \cref{sec:folklore-sqrt-fi} and \cref{sec:folklore-td}, respectively. 

    \subsection{Squared Fidelity} \label[appendix]{sec:folklore-sqr-fi}

    We present two quantum algorithms for pure-state squared fidelity estimation in \cref{lemma:folklore-sqr-fi}, where one is for query complexity and the other is for sample complexity. 

    \begin{lemma} [Folklore approaches for pure-state squared fidelity estimation] \label{lemma:folklore-sqr-fi} \
    \begin{itemize}
        \item We can estimate the squared fidelity between two pure quantum states to within additive error $\varepsilon$ with probability at least $2/3$ using $O\rbra{1/\varepsilon^2}$ independent samples of them. 
        \item There is a quantum query algorithm that estimates the squared fidelity between two pure quantum states to within additive error $\varepsilon$ with probability at least $2/3$ with query complexity $O\rbra{1/\varepsilon}$. 
    \end{itemize}
    \end{lemma}
    \begin{proof}
        The quantum circuit for estimating the squared fidelity between pure quantum states $\ket{\varphi}$ and $\ket{\psi}$ is given in \cref{fig:squared-fi}, which is directly based on the SWAP test \cite{BCWdW01}.
        \begin{figure} [!htp]
        \centering
        \begin{quantikz} [row sep = {20pt, between origins}]
            \lstick{$\ket{0}$} & \gate{H} & \ctrl{2} & \gate{H} & \meter{} & \setwiretype{c} & \rstick{$x \in \cbra{0, 1}$} \\
            \lstick{$\ket{\varphi}$} & \qw & \swap{1} & \qw & \qw \\
            \lstick{$\ket{\psi}$} & \qw & \targX{} & \qw & \qw \\
        \end{quantikz}
        \caption{Quantum circuit for estimating the squared fidelity.}
        \label{fig:squared-fi}
        \end{figure}
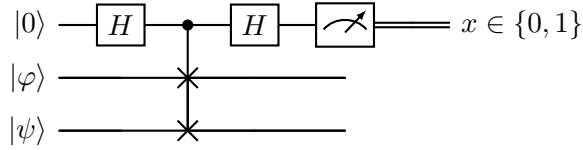
        
        It can be shown that 
        \begin{equation}
        \Pr \sbra{ x = 0 } = \frac{1 + \mathrm{F}^2\rbra{\ket{\varphi}, \ket{\psi}}}{2}.
        \end{equation}
        By repeating this process $O\rbra{1/\varepsilon^2}$ times, we can estimate $\mathrm{F}^2\rbra{\ket{\varphi}, \ket{\psi}}$ to within additive error $\varepsilon$ with probability at least $2/3$. 

        Now suppose that $U_\varphi$ and $U_\psi$ prepare $\ket{\varphi}$ and $\ket{\psi}$, respectively. 
        Similar to \cref{fig:squared-fi}, we construct the quantum circuit $U$ in \cref{fig:squared-fi-query}, using $1$ query to each of $U_\varphi$ and $U_\psi$.
        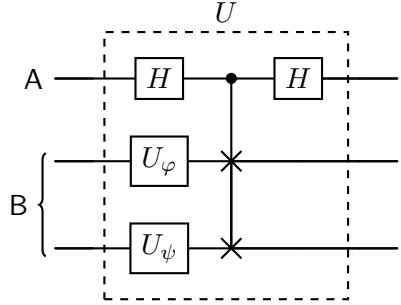
\begin{figure} [!htp]
        \centering
        \begin{quantikz}
            \lstick[1]{$\mathsf{A}$} & \qw & \gate{H} \gategroup[3,steps=3,style={dashed,inner sep=6pt}]{$U$} & \ctrl{2} & \gate{H} & \qw & \qw \\
            \lstick[2]{$\mathsf{B}$} & \qw & \gate{U_{\varphi}} & \swap{1} & \qw & \qw & \qw \\
            & \qw & \gate{U_{\psi}} & \targX{} & \qw & \qw & \qw
        \end{quantikz}
        \caption{Quantum circuit for estimating the squared fidelity with state-preparation oracles.}
        \label{fig:squared-fi-query}
        \end{figure}
        
        It can be shown that 
        \begin{equation}
            U\ket{0}_{\mathsf{A}}\ket{0}_{\mathsf{B}} = \sqrt{p} \ket{0}_{\mathsf{A}} \ket{\phi_0}_{\mathsf{B}} + \sqrt{1-p} \ket{1}_{\mathsf{A}} \ket{\phi_1}_{\mathsf{B}},
        \end{equation}
        where $\ket{\phi_0}$ and $\ket{\phi_1}$ are some normalized pure quantum states, and 
        \begin{equation}
            p = \frac{1 + \mathrm{F}^2\rbra{\ket{\varphi}, \ket{\psi}}}{2}.
        \end{equation}
        By \cref{thm:amplitude-estimation} (quantum amplitude estimation), we can obtain an estimate $\tilde x$ of $p$ to within additive error $\varepsilon/2$ using $O\rbra{1/\varepsilon}$ queries to $U$. 
        Then, $2\tilde x - 1$ is an $\varepsilon$-estimate of $\mathrm{F}^2\rbra{\ket{\varphi}, \ket{\psi}}$ to within additive error $\varepsilon$. 
        The proof completes by noting that one query to $U$ consists of one query to each of $U_{\varphi}$ and $U_{\psi}$.
    \end{proof}

    \subsection{Square Root Fidelity} \label[appendix]{sec:folklore-sqrt-fi}

    We present two quantum algorithms for pure-state square root fidelity estimation in \cref{lemma:folklore-sqrt-fi}, which are directly obtained by the quantum algorithms given in \cref{lemma:folklore-sqr-fi}. 

    \begin{lemma} [Folklore approaches for pure-state square root fidelity estimation] \label{lemma:folklore-sqrt-fi} \
    \begin{itemize}
        \item We can estimate the square root fidelity between two pure quantum states to within additive error $\varepsilon$ with probability at least $2/3$ using $O\rbra{1/\varepsilon^4}$ independent samples of them. 
        \item There is a quantum query algorithm that estimates the square root fidelity between two pure quantum states to within additive error $\varepsilon$ with probability at least $2/3$ with query complexity $O\rbra{1/\varepsilon^2}$. 
    \end{itemize}
    \end{lemma}
    \begin{proof}
        This can be done by estimating the squared fidelity between the two pure quantum states to within additive error $\varepsilon^2$ by \cref{lemma:folklore-sqr-fi} and then taking the square root of the estimated value. 
        The numerical error is guaranteed by \cref{prop:sqrt-stable}.
    \end{proof}

    \subsection{Trace Distance} \label[appendix]{sec:folklore-td}

    The folklore approaches (for both query complexity and sample complexity) for pure-state trace distance estimation were already formally presented in \cite[Theorem A.2]{WZ23}. 
    Here, we include them for completeness and with a more systematic and simpler proof. 

    \begin{lemma} [Folklore approaches for pure-state trace distance estimation] \label{lemma:folklore-td} \
    \begin{itemize}
        \item We can estimate the square root fidelity between two pure quantum states to within additive error $\varepsilon$ with probability at least $2/3$ using $O\rbra{1/\varepsilon^4}$ independent samples of them. 
        \item There is a quantum query algorithm that estimates the square root fidelity between two pure quantum states to within additive error $\varepsilon$ with probability at least $2/3$ with query complexity $O\rbra{1/\varepsilon^2}$. 
    \end{itemize}
    \end{lemma}
    \begin{proof}
        This can be done by obtaining an estimate $\tilde x$ of the squared fidelity between the two pure quantum states to within additive error $\varepsilon^2$ by \cref{lemma:folklore-sqr-fi} and then returning $\sqrt{1 - \tilde x}$ as the estimate of trace distance. 
        The correctness is due to \cref{eq:relation-td-fi} and the numerical error is guaranteed by \cref{prop:sqrt-stable}.
    \end{proof}

    \section{Sample Lower Bounds} \label[appendix]{sec:qslb}

    In this appendix, we prove sample lower bounds for estimating the trace distance and square root fidelity between pure quantum states in \cref{sec:slb-td} and \cref{sec:slb-fi}, respectively. 

    To prove the lower bounds, we need the Holevo-Helstrom bound for quantum state discrimination \cite{Hel67,Hol73}.
    Here, we use the version given in \cite{Wil13}.

    \begin{theorem} [Quantum state discrimination, cf.\ {\cite[Section 9.1.4]{Wil13}}] \label{thm:holevo}
        Suppose that $\rho_0$ and $\rho_1$ are two quantum states. 
        Let $\varrho$ be a random quantum state such that $\varrho = \rho_0$ or $\varrho = \rho_1$ with equal probability. 
        By measuring $\varrho$ according to a positive operator-valued measure (POVM) $\Lambda = \cbra{\Lambda_0, \Lambda_1}$,
        the success probability of distinguishing the two cases is bounded by
        \begin{equation}
             p_{\textup{succ}} = \frac 1 2 \tr\rbra*{\Lambda_0\rho_0} + \frac 1 2 \tr\rbra*{\Lambda_1\rho_1} \leq \frac{1}{2}\rbra*{1+\mathrm{T}\rbra{\rho_0, \rho_1}}. 
        \end{equation}
    \end{theorem}

    \subsection{Trace Distance} \label[appendix]{sec:slb-td}

    We give a quantum sample lower bound for pure-state trace distance estimation, which is inspired by the proof of \cref{thm:qlower-td} for the quantum query lower bound for pure-state trace distance estimation. 

    \begin{theorem} [Sample lower bounds for pure-state trace distance estimation] \label{thm:slb-td}
        For $\varepsilon \in \rbra{0, 1/2}$, estimating the trace distance between two pure quantum states to within additive error $\varepsilon$ with probability at least $2/3$ requires $\Omega\rbra{1/\varepsilon^2}$ independent samples of them.
    \end{theorem}
    \begin{proof}
        We consider two $n$-dimensional pure quantum states $\ket{\psi^\pm}$ where $n$ is even and 
        \begin{equation}
        \ket{\psi^\pm} = \sum_{j \in \sbra{n}} \sqrt{\frac{1 \pm \rbra{-1}^j 2 \varepsilon}{n}} \ket{j}.
        \end{equation}
        The trace distance between $\ket{\psi^+}$ and $\ket{\psi^-}$ is 
        \begin{align}
        \mathrm{T}\rbra{\ket{\psi^+}, \ket{\psi^-}} & = \sqrt{1 - \mathrm{F}\rbra{\ket{\psi^+}, \ket{\psi^-}}^2} \\
        & = \sqrt{1 - \rbra{1 - 4\varepsilon^2}} = 2 \varepsilon. \label{eq:psi-pm-td}
        \end{align}
        Suppose that we can estimate the trace distance between any two pure quantum states to within additive error $\varepsilon$ with probability at least $2/3$ using $S$ independent samples of each of them. 
        Then, the way we estimate the trace distance to within additive error $\varepsilon$ actually implies a quantum hypothesis testing experiment for distinguishing $\ket{\psi^+}^{\otimes S}$ and $\ket{\psi^-}^{\otimes S}$. 
        To see this, let $\ket{\Psi} = \ket{\psi}^{\otimes S}$ be a pure quantum state such that $\ket{\Psi} = \ket{\psi^+}^{\otimes S}$ or $\ket{\Psi} = \ket{\psi^-}^{\otimes S}$ with equal probability. Then, we consider the following process. 
        \begin{enumerate}
            \item Let $d$ be the estimate of the trace distance between $\ket{\psi}$ and $\ket{\psi^+}$ to within additive error $\varepsilon$, which is obtained using $S$ independent samples of $\ket{\psi}$ and $\ket{\psi^+}$, i.e., $\ket{\Psi}$ and $\ket{\psi^+}^{\otimes S}$.
            \item If $d < \varepsilon$, then return that $\ket{\Psi} = \ket{\psi^+}^{\otimes S}$; and return that $\ket{\Psi} = \ket{\psi^-}^{\otimes S}$ otherwise.
        \end{enumerate}
        To see the correctness of the above process, we consider two cases $\ket{\Psi} = \ket{\psi^+}^{\otimes S}$ and $\ket{\Psi} = \ket{\psi^-}^{\otimes S}$ separately as follows. 
        \begin{itemize}
            \item $\ket{\Psi} = \ket{\psi^+}^{\otimes S}$. In this case, the trace distance between $\ket{\psi}$ and $\ket{\psi^+}$ is $0$. Therefore, it holds that $d < \varepsilon$ with probability at least $2/3$. 
            According to Step 2 of the process, it will return $\ket{\Psi} = \ket{\psi^+}^{\otimes S}$ with probability at least $2/3$. 
            \item $\ket{\Psi} = \ket{\psi^-}^{\otimes S}$. In this case, by \cref{eq:psi-pm-td}, the trace distance between $\ket{\psi}$ and $\ket{\psi^+}$ is
            \begin{equation}
                \mathrm{T}\rbra{\ket{\psi}, \ket{\psi^+}} = \mathrm{T}\rbra{\ket{\psi^+}, \ket{\psi^-}} = 2\varepsilon.
            \end{equation}
            Therefore, it holds that $d > \varepsilon$ with probability at least $2/3$. 
            According to Step 2 of the process, it will return $\ket{\Psi} = \ket{\psi^-}^{\otimes S}$ with probability at least $2/3$. 
        \end{itemize}
        As shown above, the process can distinguish $\ket{\psi^+}^{\otimes S}$ and $\ket{\psi^-}^{\otimes S}$ with probability
        \begin{equation} \label{eq:p-succ-lower-bound}
            p_{\textup{succ}} \geq \frac 2 3.
        \end{equation}
        By \cref{thm:holevo}, it holds that
        \begin{align}
        p_{\textup{succ}} 
        & \leq \frac 1 2 \rbra*{1 + \mathrm{T}\rbra*{\ket{\psi^+}^{\otimes S}, \ket{\psi^-}^{\otimes S}}} \\
        & = \frac 1 2 \rbra*{ 1 + \sqrt{1 - \mathrm{F}\rbra*{\ket{\psi^+}^{\otimes S}, \ket{\psi^-}^{\otimes S}}^2} } \\
        & = \frac 1 2 \rbra*{ 1 + \sqrt{1 - \mathrm{F}\rbra*{\ket{\psi^+}, \ket{\psi^-}}^{2S}} } \\
        & = \frac 1 2 \rbra*{ 1 + \sqrt{1 - \rbra*{1 - \mathrm{T}\rbra*{\ket{\psi^+}, \ket{\psi^-}}^2}^{S}} } \\
        & = \frac 1 2 \rbra*{ 1 + \sqrt{1 - \rbra*{1 - 4\varepsilon^2}^{S}} }. \label{eq:p-succ-upper-bound}
        \end{align}
        By \cref{eq:p-succ-lower-bound} and \cref{eq:p-succ-upper-bound}, we have
        \begin{equation}
            \frac{8}{9} \geq \rbra*{1 - 4\varepsilon^2}^S \geq 1 - 4S\varepsilon^2,
        \end{equation}
        which implies that $S \geq 1/\rbra{36\varepsilon^2} = \Omega\rbra{1/\varepsilon^2}$. 
    \end{proof}

    \subsection{Squared Fidelity}

    An $\Omega\rbra{1/\varepsilon^2}$ sample lower bound for pure-state squared fidelity estimation was already given in Lemma 13 of the full version of \cite{ALL22}. 
    Here, we include it for completeness and our proof is inspired by the proof for pure-state trace distance estimation (see \cref{thm:slb-td}), which is analogous to \cref{thm:qlower-sqr-fi} as for \cref{thm:qlower-td}.

    \begin{theorem} [Sample lower bounds for pure-state squared fidelity estimation] \label{thm:slb-sqr-fi}
        For $\varepsilon \in \rbra{0, 1/2}$, estimating the squared fidelity between two pure quantum states to within additive error $\varepsilon$ with probability at least $2/3$ requires $\Omega\rbra{1/\varepsilon^2}$ independent samples of them.
    \end{theorem}
    \begin{proof}
        The proof uses the same choice of the pure quantum states $\ket{\psi^+}$ and $\ket{\psi^-}$ as that of \cref{thm:slb-td}.
        In the following, we will use the same notations as in the proof of \cref{thm:slb-td}.
        Suppose that we can estimate the squared fidelity between any two pure quantum states to within additive error $\varepsilon$ with probability at least $2/3$ using $S$ independent samples of each of them. 
        Then, the way we estimate the squared fidelity to within additive error $\varepsilon$ actually implies a quantum hypothesis testing experiment for distinguishing $\ket{\psi^+}^{\otimes S}$ and $\ket{\psi^-}^{\otimes S}$. 
        To see this, let $\ket{\Psi} = \ket{\psi}^{\otimes S}$ be a pure quantum state such that $\ket{\Psi} = \ket{\psi^+}^{\otimes S}$ or $\ket{\Psi} = \ket{\psi^-}^{\otimes S}$ with equal probability. Then, we consider the following process. 
        \begin{enumerate}
            \item Let $d$ be the estimate of the squared fidelity between $\ket{\psi}$ and $\ket{\phi}$ to within additive error $\varepsilon$, which is obtained using $S$ independent samples of $\ket{\psi}$ and $\ket{\phi}$, where
            \begin{equation}
                \ket{\phi} = \sqrt{\frac{2}{n}} \sum_{j \in \sbra{n/2}} \ket{2j}.
            \end{equation}
            \item If $d > 1/2$, then return that $\ket{\Psi} = \ket{\psi^+}^{\otimes S}$; and return that $\ket{\Psi} = \ket{\psi^-}^{\otimes S}$ otherwise.
        \end{enumerate}
        To see the correctness of the above process, we consider two cases $\ket{\Psi} = \ket{\psi^+}^{\otimes S}$ and $\ket{\Psi} = \ket{\psi^-}^{\otimes S}$ separately as follows. 
        \begin{itemize}
            \item $\ket{\Psi} = \ket{\psi^+}^{\otimes S}$. In this case, the squared fidelity between $\ket{\psi}$ and $\ket{\phi}$
            \begin{equation}
                \mathrm{F}^2\rbra{\ket{\psi^+}, \ket{\phi}} = \frac 1 2 + \varepsilon.
            \end{equation}
            Therefore, it holds that $d > 1/2$ with probability at least $2/3$. 
            According to Step 2 of the process, it will return $\ket{\Psi} = \ket{\psi^+}^{\otimes S}$ with probability at least $2/3$. 
            \item $\ket{\Psi} = \ket{\psi^-}^{\otimes S}$. In this case, by \cref{eq:psi-pm-td}, the squared fidelity between $\ket{\psi}$ and $\ket{\phi}$ is
            \begin{equation}
                \mathrm{F}^2\rbra{\ket{\psi^-}, \ket{\phi}} = \frac 1 2 - \varepsilon.
            \end{equation}
            Therefore, it holds that $d < 1/2$ with probability at least $2/3$. 
            According to Step 2 of the process, it will return $\ket{\Psi} = \ket{\psi^-}^{\otimes S}$ with probability at least $2/3$. 
        \end{itemize}
        As shown above, the process can distinguish $\ket{\psi^+}^{\otimes S}$ and $\ket{\psi^-}^{\otimes S}$ with probability $p_{\textup{succ}} \geq 2/3$.
        Therefore, following the proof of \cref{thm:slb-td}, we will obtain the same sample lower bound that $S = \Omega\rbra{1/\varepsilon^2}$. 
    \end{proof}

    \subsection{Square Root Fidelity} \label[appendix]{sec:slb-fi}

    As a corollary of \cref{thm:slb-sqr-fi}, we prove a quantum sample lower bound for pure-state square root fidelity estimation in \cref{thm:slb-sqrt-fi}. 

    \begin{theorem} [Sample lower bounds for pure-state square root fidelity estimation] \label{thm:slb-sqrt-fi}
        For $\varepsilon \in \rbra{0, 1/4}$, estimating the square root fidelity between two pure quantum states to within additive error $\varepsilon$ with probability at least $2/3$ requires $\Omega\rbra{1/\varepsilon^2}$ independent samples of them.
    \end{theorem}
    \begin{proof}
        The proof is similar to that of \cref{thm:qlb-sqrt-fi}. 
        Here, we just note that for an estimate $\tilde x$ of the square root fidelity to within additive error $\varepsilon/2$, $\tilde x^2$ can be used as an estimate of the squared fidelity to within additive error $\varepsilon$. 
    \end{proof}
    
\end{document}